\documentclass[11pt]{article}
\usepackage[margin=1in]{geometry}
\usepackage{amsmath,amssymb,amsthm,mathtools,bm}
\usepackage{booktabs}
\usepackage{enumitem}
\usepackage{natbib}
\usepackage{hyperref}
\hypersetup{colorlinks=true,linkcolor=blue,citecolor=blue,urlcolor=blue}
\usepackage{algorithm}
\usepackage{algpseudocode}
\usepackage{graphicx} % For tables
\usepackage{authblk}   
\usepackage{multirow}

\title{\textbf{Adaptive Penalized Doubly Robust Regression for Longitudinal Data}}
\author[1]{Yuyao Wang}
\author[1]{Yu Lu}
\author[1]{Tianni Zhang}
\author[1]{Mengfei Ran\footnote{Corresponding author. \url{mengfei.ran@xjtlu.edu.cn}}}
\affil[1]{Wisdom Lake Academy of Pharmacy, Xi'an Jiaotong-Liverpool University}
\date{}

\theoremstyle{plain}
\newtheorem{theorem}{Theorem} % Number theorems within sections
\newtheorem{proposition}{Proposition}
\newtheorem{lemma}{Lemma}
\theoremstyle{definition}
\newtheorem{assumption}{Assumption}

\theoremstyle{remark}
\newtheorem{remark}{Remark}

\newcommand{\R}{\mathbb{R}}
\newcommand{\E}{\mathbb{E}}
\newcommand{\Var}{\mathrm{Var}}
\newcommand{\Prob}{\mathbb{P}}
\newcommand{\Cov}{\text{Cov}}
\newcommand{\ind}{\mathbb{I}}

\newcommand{\abs}[1]{\left\lvert #1\right\rvert}

\DeclareMathOperator{\Tr}{Tr}

\begin{document}
\maketitle

\begin{abstract}
Longitudinal data often involve heterogeneity, sparse signals, and contamination from response outliers or high-leverage observations especially in biomedical science. Existing methods usually address only part of this problem, either emphasizing penalized mixed-effects modeling without robustness or robust mixed-effects estimation without high-dimensional variable selection. We propose a doubly adaptive robust regression (DAR-R) framework for longitudinal linear mixed-effects models. It combines a robust pilot fit, doubly adaptive observation weights for residual outliers and leverage points, and folded-concave penalization for fixed-effect selection, together with weighted updates of random effects and variance components. We develop an iterative reweighting algorithm and establish estimation and prediction error bounds, support recovery consistency, and oracle-type asymptotic normality. Simulations show that DAR-R improves estimation accuracy, false-positive control, and covariance estimation under both vertical outliers and bad leverage contamination. In the TADPOLE/ADNI Alzheimer’s disease application, DAR-R achieves accurate and stable prediction of ADAS13 while selecting clinically meaningful predictors with strong resampling stability.
\end{abstract}

\section{Introduction}

Longitudinal and clustered data arise routinely in biomedical research, economics, and the social sciences, where subjects are measured repeatedly over time and the resulting within-subject correlation must be accounted for in valid inference and efficient prediction.
Linear mixed-effects models (LMMs) provide a flexible and widely used framework by combining population-level fixed effects with subject-specific random effects; the foundational formulation for longitudinal data was developed by \cite{laird1982}. Comprehensive treatments and mainstream implementations are now standard in longitudinal analysis \citep{pinheiro2000, verbeke2000, diggle2002}, with widely used software such as \texttt{lme4} described by \cite{bates2015}. At the same time, modern longitudinal studies often collect a large number of candidate predictors, making variable selection for mixed models an essential component of scientific discovery.

A growing literature addresses sparse estimation and variable selection for mixed-effects models through penalized likelihood and related regularization strategies. Early work considers joint selection of fixed and random effects \citep{bondell2010, ibrahim2011}. In the high-dimensional regime, $\ell_1$-regularized estimators for LMMs have been analyzed with provable guarantees \citep{schelldorfer2011}. Despite these advances, most existing procedures are primarily developed under clean-data assumptions and may lose stability or selection accuracy when response outliers or covariate leverage points are present.

In practice, longitudinal data are frequently corrupted by outliers in responses (e.g., sporadic assay failures, recording errors) and by leverage points in covariates (e.g., data entry mistakes, device miscalibration, atypical covariate configurations). Robust statistics provides classical tools for downweighting aberrant observations \citep{huber1964, hampel1986, huber2009}, and robust mixed-model procedures have been explored through heavy-tailed modeling and robust likelihood/estimating-equation approaches. For instance, robust REML-type ideas for mixed models were investigated by \cite{richardsonwelsh1995}, and $t$-based robustification for mixed-effects models was studied by \cite{pinheiroliu2001}. Nevertheless, most existing robust mixed-model methods focus on low-dimensional settings, and do not directly address high-dimensional variable selection with simultaneous protection against response outliers and covariate leverage points \citep{koller2016}. Relatedly, \cite{ranxu2024} proposed robust variable selection via nonconcave penalties with an upgraded parsimonious dynamic covariance modeling, highlighting the value of combining robustness, dependence modeling, and sparsity in correlated-data settings. Here we use the term doubly robust in the outlier-robustness sense, emphasizing protection against both response contamination and covariate leverage (not the causal-inference meaning).

This paper develops an adaptive penalized doubly robust regression framework for  longitudinal analysis under high-dimensionality, termed longitudinal DAR-R. The method integrates three ingredients.  First, robust pilot fits deliver reliable residuals and scale estimates that serve as anchors for calibration. Second, DAR-R assigns observation-specific weights that reflect two complementary notions of outlyingness: one component downweights large standardized residuals, while the other targets covariate leverage through robust distance measures \citep{rousseeuw1987, rousseeuw1999mcd}. These two components are linked by a global discrepancy factor that adapts the overall aggressiveness of downweighting to the contamination level in the data, so the procedure remains efficient for clean samples yet protective under severe contamination. Third, after weighting, fixed effects are estimated using folded-concave regularization: \cite{fanli2001} established the oracle-property viewpoint for nonconcave penalties, and \cite{zhang2010} proposed the MCP family. Random effects and variance components are then updated through a weighted empirical Bayes/REML step, leading to a stable iterative reweighting algorithm.  The theoretical analysis builds on the general M-estimation framework of \cite{negahban2012} together with modern results for nonconvex regularization \citep{loh2015, wainwright2019}.

Our main contributions are as follows. We introduce longitudinal DAR-R, a robust-and-sparse framework for high-dimensional LMMs that simultaneously protects against response outliers and covariate leverage through doubly adaptive weighting with a data-adaptive discrepancy factor. We also provide a practical and scalable implementation via an iterative reweighting algorithm that couples weighted empirical Bayes/REML updates for random effects and variance components with coordinate-descent style updates for folded-concave penalization \citep{breheny2011}. On the theoretical side, we establish nonasymptotic bounds for estimation and prediction, and further show consistent support recovery and oracle asymptotic normality under suitable regularity conditions. Finally, extensive simulations and an application study illustrate that DAR-R achieves improved robustness and variable-selection accuracy relative to existing penalized mixed-model competitors, particularly under mixed contamination scenarios.

The rest of the paper is organized as follows. Section~\ref{sec:method} presents the longitudinal DAR-R model formulation and the doubly adaptive weighting scheme. Section~\ref{sec:algo} describes the computational procedure and iterative reweighting algorithm. Section~\ref{sec:theory} establishes the theoretical properties of the proposed estimator. Section~\ref{sec:simulation} reports simulation studies, and Section~\ref{sec:application} provides a real-data application to the TADPOLE/ADNI longitudinal Alzheimer’s disease cohort. Section~\ref{sec:conclusion} concludes with a summary and future directions.

%%%%%%%%%%%%%%%%%%%%%%%%%%%%%%%%%%%%%%%%%%%%%%%%%%%%%%%%%%%%%%%%%%%%%%%%
\section{Methodology}
\label{sec:method}

\subsection{Model Specification}
For subjects $i=1,\dots,n$ and times $t=1,\dots,T_i$, let $Y_{it}\in\R$ denote the outcome, with fixed-effect covariates $X_{it}\in\R^p$ and random-effect covariates $Z_{it}\in\R^q$. We consider the longitudinal linear mixed-effects model
\begin{equation}\label{eq:model}
Y_{it}=X_{it}^\top \beta^\star + Z_{it}^\top b_i + \varepsilon_{it},
\end{equation}
where $\beta^\star\in\R^p$ is the unknown fixed-effect vector of interest, $b_i\in\R^q$ is a subject-specific random effect, and $\varepsilon_{it}$ is the idiosyncratic noise. We assume $\E(b_i)=0$ and $\Cov(b_i)=D^\star$, while $\E(\varepsilon_{it})=0$ and $\Var(\varepsilon_{it})=(\sigma^\star)^2$, with independence across subjects and $(i,t)$ conditional on $(X_{it},Z_{it})$. The within-subject dependence is captured through the shared random effect $b_i$, following the standard mixed-model formulation; see, e.g., \citep{laird1982, verbeke2000, diggle2002}.

For compact notation, define the stacked response and design matrices $Y_i=(Y_{i1},\dots,Y_{iT_i})^\top$, $X_i=\big(X_{i1},\dots,X_{iT_i}\big)^\top\in\R^{T_i\times p}$ and $Z_i=\big(Z_{i1},\dots,Z_{iT_i}\big)^\top\in\R^{T_i\times q}$, so that $Y_i=X_i\beta^\star+Z_i b_i+\varepsilon_i$. Let $N=\sum_{i=1}^n T_i$ be the total number of observations. Our primary goal is estimation and variable selection for $\beta^\star$ when $p$ can be comparable to or much larger than $N$. We assume $\beta^\star$ is sparse with support $S=\{j:\beta_j^\star\neq 0\}$ satisfying $|S|=s\ll p$.

\subsection{Doubly Adaptive Objective}
Penalized mixed-model estimation can be sensitive to contamination because atypical observations may appear either as aberrant responses (large residuals) or as covariate leverage points (extreme $X$-profiles). Our strategy is to build observation-specific weights that downweight both sources of contamination, while automatically adapting the overall strength of downweighting to the ``cleanliness'' of the dataset.

Let $(\tilde\beta,\{\tilde b_i\},\tilde\sigma)$ be robust pilot estimates for the mixed model. Such pilots can be obtained via robust REML-type ideas \citep{richardsonwelsh1995}, $t$-based robustification \citep{pinheiroliu2001, lange1989}, or practical robust mixed-model implementations \citep{koller2016}. To assess covariate leverage robustly, we also compute a robust location--scatter pair $(\tilde\ell,\tilde\Sigma)$ for $\{X_{it}\}$, for example using the Minimum Covariance Determinant (MCD) estimator \citep{rousseeuw1999mcd}, which is standard in robust outlier detection \citep{rousseeuw1987, maronna2006}.

Based on the pilot fit, we quantify response outlyingness by the standardized residual
\[
\tilde r_{it}
=
\frac{Y_{it}-X_{it}^\top\tilde\beta-Z_{it}^\top\tilde b_i}{\tilde\sigma},
\]
so that unusually large $\abs{\tilde r_{it}}$ indicates an aberrant outcome relative to the fitted mixed model. In parallel, we measure covariate leverage using a robust Mahalanobis-type distance computed from $(\tilde\ell,\tilde\Sigma)$,
\[
d^2(X_{it})
=
(X_{it}-\tilde\ell)^\top\tilde\Sigma^{-1}(X_{it}-\tilde\ell),
\]
where larger values correspond to atypical covariate profiles in the fixed-effect design space. Here $\tilde\ell\in\R^p$ denotes a robust estimate of the location (center) of the covariate vector $X_{it}$, paired with a robust scatter estimate $\tilde\Sigma$ (e.g., from MCD).

To adapt the aggressiveness of downweighting to the overall contamination level, we introduce a global discrepancy factor based on the empirical distribution of $\abs{\tilde r_{it}}$. Let
\[
F_n^+(u)=\frac{1}{N}\sum_{i=1}^n\sum_{t=1}^{T_i}\mathbf 1\!\left(\abs{\tilde r_{it}}\le u\right),\qquad 
F_0^+(u)=\Prob(\abs{Z}\le u),~ Z\sim N(0,1),
\]
and define
\[
\tilde\delta=\sup_{u\ge 0}\big|F_n^+(u)-F_0^+(u)\big|.
\]
Intuitively, $\tilde\delta$ is small when standardized residuals behave close to the nominal reference, and becomes larger when the residual distribution exhibits noticeable tail inflation or distortion; it therefore calibrates how aggressively the procedure should downweight. If a different nominal error shape is more appropriate, $F_0^+$ can be replaced accordingly without changing the rest of the construction.

We then define doubly adaptive weights $w_{it}\in[0,1]$ as
\begin{equation}\label{eq:weights}
w_{it}
=
\phi_1\!\big(\tilde\delta\,\abs{\tilde r_{it}}\big)\cdot 
\phi_2\!\big(\tilde\delta\, d^2(X_{it})\big),
\end{equation}
where $\phi_1,\phi_2:[0,\infty)\to[0,1]$ are non-increasing functions with $\phi_k(0)=1$ and $\phi_k(u)\to 0$ as $u\to\infty$. This form follows the bounded-influence spirit of robust M-estimation \citep{huber1964, hampel1986, huber2009}, while explicitly separating response outlyingness and covariate leverage. Practically, two convenient choices are Huber-type weight $\phi(u)=\min(1,c/u)$ and Tukey bisquare weight $\phi(u)=\{1-(u/c)^2\}^2\mathbf 1(u\le c)$; both yield $w_{it}\approx 1$ for typical observations and substantially smaller weights for extreme residuals and/or leverage points.

Given weights, we update the random effects using a robust empirical Bayes (REB) step:
\begin{equation}\label{eq:bhat}
\hat b_i(\beta)
=\arg\min_{b \in \R^q}\Big\{\sum_{t=1}^{T_i} w_{it}\big(Y_{it}-X_{it}^\top \beta - Z_{it}^\top b\big)^2 + b^\top \hat{D}^{-1} b\Big\},
\end{equation}
where $\hat D$ is a robust estimate of the random-effect covariance. Denote $W_i=\mathrm{diag}(w_{i1},\dots,w_{iT_i})$, the solution has the closed form
\[
\hat b_i(\beta)
=
(Z_i^\top W_i Z_i+\hat D^{-1})^{-1}Z_i^\top W_i (Y_i-X_i\beta).
\]
Finally, we estimate the fixed effect $\beta$ by minimizing the weighted penalized criterion
\begin{equation}\label{eq:objective}
\hat\beta
=\arg\min_{\beta\in\R^p}\left\{\sum_{i=1}^n\sum_{t=1}^{T_i}
w_{it}\big(Y_{it}-X_{it}^\top\beta-Z_{it}^\top\hat b_i(\beta)\big)^2
+ N\sum_{j=1}^p p_\lambda(\abs{\beta_j})\right\},
\end{equation}
where $p_\lambda(\cdot)$ is a sparsity-inducing penalty. We focus on folded-concave penalties such as SCAD and MCP, because \cite{fanli2001} established their oracle-type behavior and \cite{zhang2010} proposed the MCP family, both of which help reduce shrinkage bias on strong signals compared to convex $\ell_1$ penalties. Computational details for optimizing \eqref{eq:objective} and updating $(\hat D,\hat\sigma)$ are deferred to Section~\ref{sec:algo}.
%%%%%%%%%%%%%%%%%%%%%%%%%%%%%%%%%%%%%%%%%%%%%%%%%%%%%%%%%%%%%%%%%%%%%%%%
\section{Computation}
\label{sec:algo}

The optimization problem in \eqref{eq:objective} is inherently coupled across multiple layers of unknowns, so a one-shot minimization is not practical. The fixed effects $\beta$ determine fitted values and hence the standardized residuals that enter the doubly adaptive weights $w_{it}$; given these weights, the random effects $\hat b_i(\beta)$ depend on $\beta$ through the weighted REB update \eqref{eq:bhat}; and the weights are then recalibrated from the current fit through the residual- and leverage-based outlyingness measures. In addition, the variance components $(D,\sigma^2)$ are naturally tied to the current estimates of $\{b_i\}$ and the weighted residual scale. This circular dependence suggests an alternating procedure that repeatedly updates auxiliary quantities such as weights, random effects, and variance components given the current $\beta$, and updates $\beta$ by solving a weighted penalized subproblem with those auxiliary quantities held fixed. The resulting scheme can be interpreted as an EM-/ECM-type approach \citep{dempster1977, meng1993} combined with an iterative reweighting (IRLS) viewpoint: each iteration refreshes the data-adaptive weights to attenuate response outliers and covariate leverage points, and then performs a stabilized penalized mixed-model update under the refreshed weights.

We alternate between an E-step, where we update the unobserved components (weights and random effects) given the current parameters, and an M-step, where we update the parameters ($\beta$) given the current weights and random effects.

\subsection{EM Algorithm}
We implement the proposed estimator via an EM-style iterative reweighting procedure. At iteration $k$, suppose we have current estimates $(\tilde\beta^{(k)}, \tilde D^{(k)}, \tilde\sigma^{(k)})$. The next iterate is obtained by an E-step that refreshes the weights and random effects (and updates variance components accordingly), followed by an M-step that updates $\beta$ under the new weights.

\subsubsection{E-step: Updating Weights and Random Effects}
The E-step updates all components that are treated as latent or auxiliary in the weighted formulation. Using the current iterate $(\tilde\beta^{(k)}, \tilde D^{(k)}, \tilde\sigma^{(k)})$, we perform:

\textbf{Step 1-Update Weights:} Compute standardized residuals based on the current fit,
\begin{equation*}
\tilde r_{it}^{(k)}=\frac{Y_{it}-X_{it}^\top \tilde\beta^{(k)}-Z_{it}^\top \tilde b_i^{(k)}}{\tilde\sigma^{(k)}},    
\end{equation*}
together with the robust leverage score $d^2(X_{it})$ in Section~\ref{sec:method}. Then evaluate the global discrepancy factor $\tilde\delta^{(k)}$ and update the observation weights by
\begin{equation*}
w_{it}^{(k+1)}=\phi_1\!\big(\tilde\delta^{(k)}|\tilde r_{it}^{(k)}|\big)\cdot \phi_2\!\big(\tilde\delta^{(k)} d^2(X_{it})\big),
\end{equation*}
which is exactly the weighting rule in \eqref{eq:weights}. Intuitively, $w_{it}^{(k+1)}$ identifies and downweights observations that are currently assessed as response outliers and/or leverage points. 

\textbf{Step 2-Update Random Effects and Variance Components:} Given $w_{it}^{(k+1)}$, update the random effects by the weighted REB step \eqref{eq:bhat}:
\begin{align*}
\hat b_i^{(k+1)}(\beta) & = \left(Z_i^\top W_i^{(k+1)} Z_i+(\tilde D^{(k)})^{-1}\right)^{-1}
    Z_i^\top W_i^{(k+1)}(Y_i-X_i\beta),\\
W_i^{(k+1)} &= \mathrm{diag}(w_{i1}^{(k+1)},\dots,w_{iT_i}^{(k+1)}).    
\end{align*}
In practice, we plug in $\beta=\tilde\beta^{(k)}$ to obtain $\hat b_i^{(k+1)}=\hat b_i^{(k+1)}(\tilde\beta^{(k)})$, and define the associated posterior covariance
\[
V_i^{(k+1)}=
\left(Z_i^\top W_i^{(k+1)} Z_i+(\tilde D^{(k)})^{-1}\right)^{-1}.
\]
We then update the variance components using the following EM-type moment updates:
\begin{equation}\label{eq:D_update}
\tilde D^{(k+1)}
=
\frac{1}{n}\sum_{i=1}^n
\left\{
\hat b_i^{(k+1)}\hat b_i^{(k+1)\top}
+
V_i^{(k+1)}
\right\},
\end{equation}
and
\begin{equation}\label{eq:sigma_update}
\tilde\sigma^{2(k+1)}
=
\frac{1}{N}\sum_{i=1}^n
\left\{
\| (W_i^{(k+1)})^{1/2}\, e_i^{(k+1)} \|_2^2
+
\Tr\!\big(W_i^{(k+1)}Z_i V_i^{(k+1)} Z_i^\top\big)
\right\},
\end{equation} 
where $e_i^{(k+1)} = Y_i - X_i\tilde\beta^{(k)} - Z_i\hat b_i^{(k+1)}$. These updates yield refreshed quantities $(\tilde D^{(k+1)},\tilde\sigma^{2(k+1)})$ to be used in the subsequent M-step.

\subsubsection{M-step: Updating Fixed Effects ($\beta$)}
Conditioning on the updated weights $w_{it}^{(k+1)}$ and the refreshed variance components, the M-step updates the fixed effects by minimizing the weighted penalized objective:
\[
\hat\beta^{(k+1)} = \arg\min_{\beta\in\R^p} \left\{ \sum_{i=1}^n \left\|\left(W_i^{(k+1)}\right)^{1/2}
\big( Y_i-X_i\beta-Z_i\hat b_i^{(k+1)}(\beta)
\big) \right\|_2^2 + N\sum_{j=1}^p p_\lambda(|\beta_j|)
\right\},
\]
where $\hat b_i^{(k+1)}(\beta)$ is the closed-form REB update from \eqref{eq:bhat} using $(W_i^{(k+1)},\hat D^{(k+1)})$. This subproblem can be solved efficiently by coordinate descent for folded-concave penalties; see, e.g., \cite{breheny2011} and the general coordinate-descent literature \citep{friedman2010}. After obtaining $\hat\beta^{(k+1)}$, we set $\tilde\beta^{(k+1)}\leftarrow \hat\beta^{(k+1)}$ and proceed to the next iteration until convergence.

\subsection{Algorithm}

The complete iterative procedure is summarized in Algorithm~\ref{alg:main}.

\begin{algorithm}[!]
\caption{Iterative Algorithm for Longitudinal DAR-R}
\label{alg:main}
\begin{algorithmic}[1]
\State \textbf{Initialize:} Obtain robust pilot estimates $(\tilde\beta^{(0)}, \tilde D^{(0)}, \tilde\sigma^{(0)})$. Compute robust location/scatter $(\tilde\ell, \tilde\Sigma)$ of $\{X_{it}\}$ (e.g., via MCD). Set $k=0$.
\Repeat
    \State \textbf{Step 1 (E-step: weights).}
    \State Compute the global discrepancy factor $\tilde\delta^{(k)}=\sup_{u\ge 0}\big|F_n^+(u;\tilde\beta^{(k)})-F_0^+(u)\big|.$
    \State For all $(i,t)$, compute leverage scores $d^2(X_{it})$ and standardized residuals $\tilde r_{it}^{(k)}$.
    \State Update weights $w_{it}^{(k+1)}$ using \eqref{eq:weights} with $\tilde\delta^{(k)}$, $\tilde r_{it}^{(k)}$, and $d^2(X_{it})$.

    \State \textbf{Step 2 (E-step: random effects and variance components).}
    \State Compute $\hat b_i^{(k+1)}=\hat b_i(\tilde\beta^{(k)})$ using \eqref{eq:bhat} with weights $w_{it}^{(k+1)}$ and $\tilde D^{(k)}$.
    \State Update $\tilde D^{(k+1)}$ (and $\tilde\sigma^{2(k+1)}$ if applicable) using \eqref{eq:D_update}--\eqref{eq:sigma_update}.

    \State \textbf{Step 3 (M-step: fixed effects).}
    \State Update $\hat\beta^{(k+1)}$ by solving the weighted penalized problem in \eqref{eq:objective} (equivalently \eqref{eq:Mstep_beta}) with weights $w_{it}^{(k+1)}$ and covariance $\tilde D^{(k+1)}$.
    \State Set $\tilde\beta^{(k+1)} \leftarrow \hat\beta^{(k+1)}$ and $k \leftarrow k+1$.
\Until{convergence (e.g., $\|\tilde\beta^{(k)}-\tilde\beta^{(k-1)}\|_2/\max\{1,\|\tilde\beta^{(k-1)}\|_2\}$ is below a tolerance).}
\State \textbf{Return:} $\hat\beta = \tilde\beta^{(k)}$.
\end{algorithmic}
\end{algorithm}

The key computational challenge is Step~3 (M-step). By profiling out the random effects $\hat b_i(\beta)$ from \eqref{eq:bhat}, the objective in \eqref{eq:objective} can be rewritten as a weighted penalized least-squares problem only in $\beta$:
\[
\min_{\beta \in \R^p}
\left\{
\sum_{i=1}^n (Y_i - X_i \beta)^\top \tilde W_i (Y_i - X_i \beta)
+
N \sum_{j=1}^p p_\lambda(|\beta_j|)
\right\},
\]
where the effective weight matrix is
\[
\tilde W_i
=
W_i - W_i Z_i (Z_i^\top W_i Z_i + \hat D^{-1})^{-1} Z_i^\top W_i,
\qquad
W_i=\mathrm{diag}(w_{i1},\dots,w_{iT_i}).
\]
For nonconvex penalties such as SCAD/MCP, we adopt the local linear approximation (LLA) strategy \citep{zouli2008}, which replaces $p_\lambda(|\beta_j|)$ by a weighted $\ell_1$ term using the current iterate; the resulting subproblem is a sequence of weighted Adaptive Lasso fits. The motivation for folded-concave penalties and their oracle-type behavior follows \citep{fanli2001}.

\subsection{Coordinate Descent for the Penalized Fixed-Effect Update}
\label{ssec:cd}

At the M-step, we update $\beta$ under the refreshed weights and the current random-effect/variance estimates. To obtain a numerically stable and scalable update, we treat $\{\hat b_i^{(k+1)}\}$ as fixed within the M-step (as in an ECM implementation \citep{meng1993}) and solve a weighted penalized least-squares problem. Define the working responses
\[
Y_{it}^{\ast}
=
Y_{it}-Z_{it}^\top \hat b_i^{(k+1)},
\qquad
Y_i^\ast = Y_i - Z_i \hat b_i^{(k+1)}.
\]
Given $W_i^{(k+1)}$, the M-step update is
\begin{equation}\label{eq:Mstep_beta}
\beta^{(k+1)}
=
\arg\min_{\beta\in\R^p}
\left\{
\sum_{i=1}^n (Y_i^\ast-X_i\beta)^\top W_i^{(k+1)}(Y_i^\ast-X_i\beta)
+
N\sum_{j=1}^p p_\lambda(|\beta_j|)
\right\}.
\end{equation}
This problem has the same structure as weighted penalized regression, except that the weights are data-adaptive and refreshed across outer iterations. We solve \eqref{eq:Mstep_beta} by coordinate descent, which is efficient in high dimensions and enjoys good practical performance even for folded-concave penalties \citep{friedman2010}. Convergence properties of block/coordinate descent for non-smooth objectives have been studied in \citep{tseng2001}.

For the $j$-th coordinate update, let $x_{itj}$ be the $j$th component of $X_{it}$ and define the partial residual excluding the $j$-th coordinate,
\[
r_{it}^{(-j)} = Y_{it}^\ast - \sum_{\ell\neq j} x_{it\ell}\beta_\ell.
\]
The subproblem in $\beta_j$ becomes a one-dimensional penalized quadratic:
\[
\min_{\beta_j\in\R}
\left\{
\sum_{i=1}^n\sum_{t=1}^{T_i} w_{it}^{(k+1)}\big(r_{it}^{(-j)}-x_{itj}\beta_j\big)^2
+
N\,p_\lambda(|\beta_j|)
\right\}.
\]
Let
\[
a_j=\sum_{i=1}^n\sum_{t=1}^{T_i} w_{it}^{(k+1)}x_{itj}^2,
\qquad
b_j=\sum_{i=1}^n\sum_{t=1}^{T_i} w_{it}^{(k+1)}x_{itj}r_{it}^{(-j)}.
\]
Up to an additive constant, the subproblem is equivalent to
\[
\min_{\beta_j\in\R}\left\{\frac{1}{2}z_j(\beta_j-u_j)^2 + N\,p_\lambda(|\beta_j|)\right\},
\qquad
z_j=2a_j,\quad u_j=\frac{b_j}{a_j},
\]
so each coordinate update reduces to a scalar penalized least-squares proximal step. For SCAD and MCP, the minimizer has a closed-form thresholding rule; we implement these updates using the standard coordinate-descent formulas in \cite{breheny2011}. In particular, for MCP with concavity parameter $\gamma>1$, the update can be expressed in terms of the score $s_j=z_j u_j$ as
\[
\beta_j \leftarrow
\begin{cases}
0, & |s_j|\le N\lambda,\\[2mm]
\displaystyle \frac{\mathrm{sign}(s_j)\big(|s_j|-N\lambda\big)}{z_j-1/\gamma}, & N\lambda<|s_j|\le \gamma N\lambda z_j,\\[3mm]
\displaystyle \frac{s_j}{z_j}, & |s_j|>\gamma N\lambda z_j,
\end{cases}
\]
and the SCAD update follows an analogous piecewise form (see \cite{breheny2011}). We cycle through coordinates until the inner-loop convergence criterion is met, and then return $\beta^{(k+1)}$ to the outer iteration.

In implementation, we use warm starts across $\lambda$ values and across outer iterations, and we maintain an active-set strategy to reduce computation by prioritizing coordinates with larger empirical gradients. These standard accelerations preserve the exact objective \eqref{eq:Mstep_beta} while substantially improving scalability in high dimensions.

%%%%%%%%%%%%%%%%%%%%%%%%%%%%%%%%%%%%%%%%%%%%%%%%%%%%%%%%%%%%%%%%%%%%%%%%

\section{Theoretical Properties}
\label{sec:theory}

\subsection{Notation and Assumptions}
\label{subsec:assumptions}

Let $N=\sum_{i=1}^n T_i$ be the total number of observations and
$S=\{j:\beta_j^\star\neq 0\}$ be the true active set with $|S|=s$.
For any vector $v\in\mathbb R^p$, write $v_S$ and $v_{S^c}$ as the subvectors on $S$ and its complement.
Define the cone
\[
\mathcal C(S,3)=\Big\{\Delta\in\mathbb R^p:\ \|\Delta_{S^c}\|_1\le 3\|\Delta_S\|_1\Big\}.
\]

Throughout Section~\ref{sec:theory}, we analyze the profiled weighted loss that appears in \eqref{eq:objective}.
Given weights $W_i=\mathrm{diag}(w_{i1},\ldots,w_{iT_i})$ and a working covariance $D$,
profiling out $b_i$ leads to the effective weight matrix
\[
\tilde W_i
=
W_i - W_i Z_i\big(Z_i^\top W_i Z_i + D^{-1}\big)^{-1}Z_i^\top W_i,
\]
so that the profiled loss can be written as
\[
\mathcal L_N(\beta)
=
\frac{1}{2N}\sum_{i=1}^n (Y_i-X_i\beta)^\top \tilde W_i (Y_i-X_i\beta).
\]
Let $\hat\beta$ denote a stationary point (or a local minimizer) of
$\mathcal L_N(\beta)+\sum_{j=1}^p p_\lambda(|\beta_j|)$.
Write $\Delta=\hat\beta-\beta^\star$ and $\xi=\nabla \mathcal L_N(\beta^\star)$.

\begin{assumption}
\textbf{Clustered sampling and moments.}\label{ass_a1} 
Subjects are independent across $i$.
The within-subject size is uniformly bounded: $\max_i T_i\le T_{\max}$. Conditional on $(X_{it},Z_{it})$, we have $\E(b_i)=0$, $\E(\varepsilon_{it})=0$, $\E\|b_i\|^{2+\delta}<\infty$, and $\E|\varepsilon_{it}|^{2+\delta}<\infty$ for some $\delta>0$.
\end{assumption}

\begin{assumption}
\textbf{Design regularity.}\label{ass_a2} 
The rows of $X_i$ are sub-Gaussian with bounded second moments, and the eigenvalues of $\Sigma_X=\E(X_{it}X_{it}^\top)$ are bounded away from $0$ and $\infty$.   
\end{assumption}

\begin{assumption}
\textbf{Sparsity and growth regime.} \label{ass_a3} 
The sparsity satisfies $s\log p=o(N)$ and $s\ll p$.
\end{assumption}

\begin{assumption}
\textbf{Weight regularity (clean mass and boundedness).} \label{ass_a4}
Weights satisfy $w_{it}\in[0,1]$ for all $(i,t)$. Moreover, with probability tending to $1$, there exists a ``clean'' index set $\mathcal C\subset\{(i,t)\}$ with $|\mathcal C|\ge (1-\epsilon)N$ such that $w_{it}\ge c_w>0$ for all $(i,t)\in\mathcal C$. 
\end{assumption}

\begin{assumption}
\textbf{Penalty regularity.} \label{ass_a5}
The penalty $p_\lambda$ is either the adaptive LASSO or a folded-concave penalty (SCAD/MCP)
with $p'_\lambda(0+)=\lambda$, and there exists $a>0$ such that $p'_\lambda(t)=0$ for all $t\ge a\lambda$. The tuning parameter satisfies $\lambda\asymp \sqrt{(\log p)/N}$.   
\end{assumption}

\begin{assumption}
\textbf{Feasible--oracle score proximity.} \label{ass_a6}
Let $\mathcal L_N^{\circ}(\beta)$ be the same profiled loss as above but computed using the
population/limit versions of the weights and covariance components. Then
\[
\big\|\nabla \mathcal L_N(\beta^\star)-\nabla \mathcal L_N^{\circ}(\beta^\star)\big\|_\infty=o_p(\lambda).
\]
This condition formalizes that estimating weights and variance components does not inflate the score beyond the stochastic order of $\lambda$.    
\end{assumption}

\begin{assumption}
\textbf{Restricted strong convexity.} \label{ass_a7}
There exist constants $\alpha>0$ and $\tau\ge 0$ such that, with probability tending to $1$,
\[
\Delta^\top\Big(\nabla^2 \mathcal L_N(\beta^\star)\Big)\Delta
\ \ge\
\alpha\|\Delta\|_2^2-\tau\frac{\log p}{N}\|\Delta\|_1^2,
\qquad \forall\,\Delta\in\mathcal C(S,3).
\]   
\end{assumption}

\noindent\textbf{Remark}
Assumption \ref{ass_a1} is standard for longitudinal mixed models: independence across subjects and bounded cluster size facilitate concentration and CLT arguments. Assumption \ref{ass_a2} imposes mild tail and eigenvalue regularity on the design so that weighted score and Gram matrices concentrate uniformly. Assumption \ref{ass_a3} is the usual sparse high-dimensional regime and is mainly used to control the RSC tolerance term and yield the rate $\sqrt{s\log p/N}$. Assumption \ref{ass_a4} ensures the weighting scheme preserves a nontrivial effective sample size by keeping weights bounded away from zero on the clean majority; this is essential for both curvature and stochastic gradient bounds. Assumption \ref{ass_a5} collects standard regularity for SCAD/MCP (or adaptive Lasso) and sets $\lambda\asymp \sqrt{(\log p)/N}$. Assumption \ref{ass_a6} is a feasibility condition stating that estimating weights and variance components does not perturb the score beyond $o_p(\lambda)$. Finally, Assumption \ref{ass_a7} is the restricted strong convexity condition for the profiled weighted loss, which is compatible with  Assumption \ref{ass_a2} and \ref{ass_a4} under standard concentration arguments.

\subsection{Theoretical Results}

\begin{theorem}[Non-asymptotic bounds]\label{thm:nonasym}
Under Assumption \ref{ass_a1}--\ref{ass_a7}, assume $\lambda\ge 2\|\nabla \mathcal L_N(\beta^\star)\|_\infty$.  Then for all sufficiently large $N$, there exist constants $C_1,C_2>0$ such that
\[
\|\hat\beta-\beta^\star\|_2 \le \frac{C_1\lambda\sqrt{s}}{\alpha},
\qquad
\frac{1}{N}\| \tilde W^{1/2}X(\hat\beta-\beta^\star)\|_2^2 \le \frac{C_2\lambda^2 s}{\alpha},
\]
where $\tilde W=\mathrm{blkdiag}(\tilde W_1,\ldots,\tilde W_n)$.
\end{theorem}

\begin{remark}
The rates are of the familiar order $\lambda\sqrt{s}$ up to the curvature constant, showing that robustness mainly affects constants through the effective weighting/profiling. The condition $\lambda\ge 2\|\nabla \mathcal L_N(\beta^\star)\|_\infty$ is the standard calibration ensuring the basic inequality yields a cone constraint and hence the stated bounds.
\end{remark}

\begin{theorem}[Consistency]\label{thm:estimation}
Under Assumption \ref{ass_a1}--\ref{ass_a7} with $\lambda\asymp \sqrt{(\log p)/N}$, the estimator $\hat\beta$ from \eqref{eq:objective} satisfies:
\[
\|\hat\beta-\beta^\star\|_2=O_p\!\Big(\sqrt{\tfrac{s\log p}{N}}\Big),
\qquad
\|\hat\beta-\beta^\star\|_1=O_p\!\Big(s\sqrt{\tfrac{\log p}{N}}\Big).
\]
\end{theorem}
\begin{remark}
Theorem~\ref{thm:estimation} is an immediate consequence of Theorem~\ref{thm:nonasym} under $\lambda\asymp \sqrt{(\log p)/N}$, yielding the usual sparse rate $\sqrt{s\log p/N}$. Thus the doubly adaptive weighting does not change the first-order convergence rate when the clean-mass condition holds.
\end{remark}

\begin{theorem}[Support Recovery]\label{thm:support}
Under Assumption \ref{ass_a1}--\ref{ass_a7}, assume in addition a weighted incoherence condition holds:
\[
\big\|\big(\nabla^2_{S^cS}\mathcal L_N(\beta^\star)\big)\big(\nabla^2_{SS}\mathcal L_N(\beta^\star)\big)^{-1}\big\|_\infty
\le 1-\eta
\quad \text{for some }\eta\in(0,1),
\]
and the minimum signal satisfies $\min_{j\in S}|\beta_j^\star|\ge c_\beta \lambda$ for a sufficiently large constant $c_\beta$. Then $\Prob(\hat S=S)\to 1$, where $\hat S=\{j:\hat\beta_j\neq 0\}$.
\end{theorem}

\begin{remark}
The incoherence and beta-min conditions play the same roles as in high-dimensional regression, but they are imposed on the Hessian of the profiled weighted loss, which already accounts for within-subject dependence. Folded-concave penalties reduce shrinkage bias on strong signals, so the required beta-min constant is often milder than for convex $\ell_1$ penalties.
\end{remark}

\begin{theorem}[Oracle Asymptotic Normality]\label{thm:oracle}
Assume the conditions of Theorem~\ref{thm:support} and that $s=o(\sqrt N)$.
On the event $\{\hat S=S\}$, the active subvector $\hat\beta_S$ satisfies
\[
\sqrt{N}(\hat\beta_S-\beta_S^\star)\ \xrightarrow{d}\
N\big(0,\ G_S^{-1}\Psi_S G_S^{-1}\big),
\]
where $\nabla^2_{SS}\mathcal L_N(\beta^\star)\xrightarrow{p} G_S$ and $\Psi_S=\lim\limits_{N\to\infty}\Var\!\Big(\frac{1}{\sqrt N}\sum_{i=1}^n \nabla_S \ell_i(\beta^\star)\Big)$ when $N\to\infty$, with $\ell_i(\beta)=\frac{1}{2}\,(Y_i-X_i\beta)^\top \tilde W_i (Y_i-X_i\beta)$.
\end{theorem}

\begin{remark}
Once $\hat S=S$ and $p'_\lambda(t)=0$ on the active coordinates, the estimator behaves like an unpenalized fit restricted to $S$, leading to asymptotic normality. The limiting covariance depends on the profiled weighted loss and thus incorporates both correlation (through profiling) and robustness (through weights).
\end{remark}

\begin{proposition}[Breakdown point]\label{prop:bdp}
Assume there exists an event $\mathcal E_N$ with $\Prob(\mathcal E_N)\to 1$ such that, the pilot quantities used to construct $w_{it}$ remain bounded on $\mathcal E_N$; and on $\mathcal E_N$, all contaminated indices receive weights $w_{it}=0$ while all clean indices satisfy $w_{it}\ge c_w$. Then on $\mathcal E_N$, $\hat\beta$ depends only on the clean subsample and is insensitive to arbitrary values on the contaminated part.
\end{proposition}

\begin{remark}
This result is a breakdown-type statement: if contaminated observations receive (near-)zero weight while clean observations keep weights bounded below, then the fitted objective depends only on the clean part. We state it in this conditional form to avoid invoking a full finite-sample breakdown analysis for high-dimensional penalized mixed models.
\end{remark}

\begin{proposition}[Bounded local influence]\label{prop:if}
Consider the (unpenalized) estimating equation associated with $\mathcal L_N(\beta)$.
If $w_{it}\in[0,1]$ and the score contribution is multiplied by $w_{it}$, then the directional derivative of the estimating equation with respect to an $\epsilon$-contamination at a point $(x_0,y_0)$ is bounded by a constant multiple of $w(x_0,y_0)\|x_0\|\cdot |y_0-x_0^\top\beta^\star|$. In particular, when $w(x_0,y_0)\to 0$ for large residuals/leverage, the local influence vanishes.
\end{proposition}

\begin{remark}
The influence of a contamination point is multiplied by its weight in the score, so redescending weights make extreme residual/leverage points have negligible first-order impact. This matches the algorithmic behavior: such points are progressively downweighted across iterations.
\end{remark}

%%%%%%%%%%%%%%%%%%%%%%%%%%%%%%%%%%%%%%%%%%%%%%%%%%%%%%%%%%%%%%%%%%%%%%%%
\section{Simulation Study}
\label{sec:simulation}

To evaluate the finite-sample performance of our proposed Longitudinal DAR-R method, we conduct a comprehensive simulation study. The objective is to assess its parameter estimation accuracy, variable selection consistency and prediction performance under various data contamination scenarios, benchmarked against standard non-robust and robust-but-non-sparse methods.

\subsection{Data Generation}
We simulate data from the longitudinal linear mixed-effects model:
\[
Y_{it}=X_{it}^\top \beta^\star + Z_{it}^\top b_i + \varepsilon_{it}, \quad i=1,\dots,n, \ t=1,\dots,T_i
\]
We investigate three sample sizes, $n \in \{100, 200, 300\}$ subjects, each with $T_i=5$ repeated measurements (a balanced design), for a total of $N=500$ observations. The fixed-effects dimension is $p=200$, with a true sparsity level of $s=10$.

\begin{itemize}
    \item \textbf{Fixed Effects ($\beta^\star$):} The true coefficient vector $\beta^\star$ is sparse. The first $s=10$ elements are set to $(2, -1.5, 1, -2.5, 1.8, 2.2, -1.2, 1.5, -2, 1)$, while the remaining $p-s=190$ elements are zero.
    
    \item \textbf{Random Effects ($b_i, Z_{it}$):} We employ a simple random intercept and random slope model. The covariate matrix for random effects is $Z_{it} = (1, t/5)^\top$. The random effects $b_i = (b_{i0}, b_{i1})^\top$ are drawn from $N(0, D)$, where $D$ is a diagonal covariance matrix $D = \mathrm{diag}(\sigma_{b0}^2, \sigma_{b1}^2)$ with $\sigma_{b0}^2 = 1$ and $\sigma_{b1}^2 = 0.5^2$.
    
    \item \textbf{Covariates ($X_{it}$):} The fixed-effect covariates $X_{it} \in \R^p$ are drawn from a multivariate normal distribution $N(0, \Sigma_X)$. The covariance matrix $\Sigma_X$ has an AR(1) autoregressive structure, where $(\Sigma_X)_{jk} = \rho^{|j-k|}$ with $\rho=0.5$. This induces moderate correlation between predictors.
    
    \item \textbf{Errors ($\varepsilon_{it}$):} For clean data, the within-subject errors $\varepsilon_{it}$ are drawn independently from $N(0, \sigma_\varepsilon^2)$ with $\sigma_\varepsilon^2=1$.
\end{itemize}

\subsection{Contamination Scenarios}
We investigate three distinct scenarios with a contamination proportion of $\pi=0.1$ (10\%).

\begin{description}
    \item[S1: Clean (Baseline)] This scenario has $\pi=0$ contamination. All data are generated as described in the DGP. This baseline serves to measure the efficiency loss of our robust method; ideally, its performance should be very close to non-robust methods when no outliers are present.
    
    \item[S2: Vertical Outliers] We randomly select $\pi N = 50$ (10\%) of the total observations $(i,t)$. For these contaminated observations, we replace their error term $\varepsilon_{it}$ with a draw from $N(20, 1)$. This simulates gross errors in the response variable only, without affecting the covariate space.
    
    \item[S3: Bad Leverage Points] This is the most challenging scenario. We randomly select $\pi n = 10$ (10\%) of the $i$-th subjects. For all $T_i=5$ observations belonging to these contaminated subjects, we replace their corresponding covariates and responses:
    \begin{itemize}
        \item $X_{it, j} \sim N(10, 0.1)$ for $j \in S$ (the true active set).
        \item $Y_{it} \sim N(30, 1)$.
    \end{itemize}
    This creates clusters of observations that are severe outliers in both the $X$-space (high leverage) and the $Y$-space, making them highly influential ``bad leverage" points.
\end{description}

\subsection{Competing Methods}
We compare the performance of our proposed method, \textbf{DAR-R} (using the SCAD penalty), against three competing methods:

\begin{enumerate}
    \item \textbf{DAR-R}: Our proposed method with the SCAD penalty.
    \item \textbf{Penalized-LME (glmmLasso):} A standard non-robust method for penalized linear mixed models, implemented using the \texttt{glmmLasso} R package. This is a non-robust, sparse benchmark.
    
    \item \textbf{Robust-LME (rlmer):} A standard robust LME estimator, implemented using the \texttt{robustlmm} R package \citep{koller2016}. This method is robust but not sparse (i.e., it does not perform variable selection). We fit it on all $p$ predictors.
    
    \item \textbf{LASSO:} The standard LASSO \citep{tibshirani1996} applied to the data while ignoring the longitudinal structure (i.e., treating all $N$ observations as i.i.d.). This is a non-robust, sparse benchmark that misspecifies the data structure.

    \item \textbf{Oracle-LME:} A standard non-robust LME (e.g., via \texttt{lme4}) fit only on the true active set $S$. This represents an unattainable gold standard performance.
\end{enumerate}
For all penalized methods (DAR-R, glmmLasso, LASSO), the regularization parameter $\lambda$ is selected using 5-fold cross-validation.

\subsection{Evaluation Metrics}
For each combination of sample size and contamination scenario, we repeat the simulation $R=200$ times and summarize performance across replications. We report means with medians  because several competitors can become unstable under contamination, and the median is more robust to occasional extreme runs.

Denote $S=\{j:\beta_j^\star\neq 0\}$ as the true active set with $|S|=s=10$, and let $S^c$ denote the inactive set. We evaluate four aspects of performance: estimation accuracy for the fixed effects, variable-selection quality, random-effect covariance estimation, and out-of-sample prediction.

\textbf{Estimation Error for Fixed Effects (active vs.\ inactive coordinates).} We report two mean squared error measures for $\hat\beta$, separating the signal and noise coordinates:
    \[
    \mathrm{MSE}(S)
    =
    \frac{1}{s}\|\hat\beta_S-\beta_S^\star\|_2^2,
    \qquad
    \mathrm{MSE}(S^c)
    =
    \frac{1}{p-s}\|\hat\beta_{S^c}-\beta_{S^c}^\star\|_2^2
    =
    \frac{1}{p-s}\|\hat\beta_{S^c}\|_2^2.
    \]
$\mathrm{MSE}(S)$ measures how accurately a method recovers the true nonzero coefficients (signal recovery), and is sensitive to both contamination-induced bias and over-shrinkage. In contrast, $\mathrm{MSE}(S^c)$ measures how well the method shrinks truly null coefficients toward zero (noise suppression), and is closely related to spurious leakage into inactive variables. It separately helps distinguish methods that estimate strong signals well from methods that mainly achieve sparsity by aggressive shrinkage.

\textbf{Variable Selection Accuracy (TP/FP).} We evaluate support recovery using the numbers of true positives and false positives:
    \[
    \mathrm{TP}
    =
    \sum_{j\in S}\ind\!\big(|\hat\beta_j|>\tau\big),
    \qquad
    \mathrm{FP}
    =
    \sum_{j\notin S}\ind\!\big(|\hat\beta_j|>\tau\big),
    \]
where $\tau$ is a small numerical threshold (we use $\tau=10^{-8}$ in implementation) to avoid treating machine-level noise as a selected variable.
    
Here, $\mathrm{TP}$ is the number of truly active coefficients correctly identified (maximum $10$ in our setup), while $\mathrm{FP}$ is the number of inactive coefficients incorrectly selected. A good method should achieve a high TP and a low FP simultaneously. This pair complements $\mathrm{MSE}(S)$ and $\mathrm{MSE}(S^c)$: TP/FP captures selection behavior, whereas MSE captures estimation magnitude error.

\textbf{Random-Effect Covariance Estimation Error.} To assess estimation of the random-effects covariance matrix, we report the Frobenius norm error
    \[
    \|\hat D-D\|_{\mathrm{F}},
    \]
where $D$ is the true covariance matrix used to generate the random intercept and random slope. This metric evaluates how well each method recovers the within-subject dependence structure induced by the random effects.
    
For methods that do not explicitly estimate a mixed-model covariance structure (e.g., the i.i.d.\ LASSO benchmark), this metric is not directly applicable and is reported as missing or not available. For methods with a simplified or misspecified random-effects structure, the reported value should be interpreted with caution, since it conflates estimation error and model misspecification.

\textbf{Prediction Error.} We evaluate out-of-sample prediction on an independently generated clean test set with $n_{\text{test}}=100$ subjects and $T_i=5$ repeated measurements per subject, so $N_{\text{test}}=500$. The mean squared prediction error (MSPE) is
    \[
    \mathrm{MSPE}
    =
    \frac{1}{N_{\text{test}}}
    \sum_{i=1}^{n_{\text{test}}}
    \left\|
    Y_i^{\mathrm{new}}
    -
    X_i^{\mathrm{new}}\hat\beta
    -
    Z_i^{\mathrm{new}}\hat b_i(\hat\beta)
    \right\|_2^2,
    \]
where $\hat b_i(\hat\beta)$ denotes the predicted random effects for subject $i$ under the fitted model. Using a clean test set isolates the generalization performance of each method from contamination in the training sample. This is especially important in our setting: a robust method should not only resist outliers during estimation, but also preserve predictive accuracy on future non-contaminated data.

Overall, these metrics jointly assess the main goals of the proposed method: robust estimation of high-dimensional fixed effects, reliable support recovery, stable estimation of subject-level dependence, and accurate prediction in longitudinal settings.

\subsection{Main Results}

Tables~\ref{tab:sim_mse100_combined}--\ref{tab:sim_mse300_combined} and Tables~\ref{tab:sim_n100}--\ref{tab:sim_n300}, together with Figures~\ref{fig:mseSc_boxplot}--\ref{fig:mspe_boxplot}, provide a coherent picture of the finite-sample behavior of the competing methods under clean data (S1), vertical outliers (S2), and bad leverage contamination (S3). Across all sample sizes $n\in\{100,200,300\}$, DAR-R is the strongest overall method among the practically implementable competitors when the target is robust high-dimensional longitudinal mixed-effects estimation rather than prediction alone. In particular, DAR-R consistently exhibits near-oracle accuracy on the active coefficients, the smallest error on inactive coefficients, the best TP/FP balance for support recovery, and the smallest random-effects covariance estimation error. These conclusions are stable across the three contamination regimes and become even clearer as $n$ increases.

We first discuss fixed-effect estimation accuracy through $\mathrm{MSE}(S)$ and $\mathrm{MSE}(S^c)$, reported in Tables~\ref{tab:sim_mse100_combined}--\ref{tab:sim_mse300_combined}. The active-set error $\mathrm{MSE}(S)$ evaluates recovery of the true nonzero coefficients, whereas $\mathrm{MSE}(S^c)$ quantifies leakage into truly null coordinates. For $\mathrm{MSE}(S)$, Oracle-LME serves as an unattainable upper benchmark, and DAR-R is consistently the best practical method, remaining very close to the oracle values in all scenarios. At $n=100$, DAR-R has $\mathrm{MSE}(S)$ around $4.00$--$4.19$ (all values in these three tables are scaled by $10^3$), while other methods show larger values. The same pattern persists at $n=200$ and $n=300$: DAR-R decreases to about $1.93$--$2.01$ and then $1.27$--$1.36$, closely tracking Oracle-LME (about $1.96$--$2.08$ and then $1.28$--$1.37$), while the competing practical methods remain substantially larger. Figure~\ref{fig:mseS_boxplot} confirmed that the DAR-R boxplots are concentrated near the oracle level with relatively small dispersion, whereas Penalized-LME and Robust-LME are much higher, and Robust-LME shows visibly larger variability in several panels.

The advantage of DAR-R is even more pronounced for $\mathrm{MSE}(S^c)$, which is the most direct measure of spurious coefficient leakage in the high-dimensional setting. Across all sample sizes and all contamination scenarios, DAR-R has the smallest $\mathrm{MSE}(S^c)$ by a very large margin. At all sample sizes, DAR-R has by far the smallest $\mathrm{MSE}(S^c)$. For example, it is about $0.013$--$0.024$ at $n=100$ and $0.005$--$0.007$ at $n=300$, whereas the competing methods remain much larger (roughly $0.23$--$1.65$ across methods and scenarios). This large and persistent gap shows that DAR-R is substantially more effective at suppressing leakage into inactive coefficients. Figure~\ref{fig:mseSc_boxplot} makes this separation especially transparent on the $\log_{10}$ scale: DAR-R lies at the bottom of every panel with extremely tight spread, while all alternatives are shifted upward by one to two orders of magnitude. This is strong empirical evidence that the doubly adaptive weighting successfully protects sparse estimation from both response contamination and leverage contamination.

The support recovery results in Tables~\ref{tab:sim_n100}--\ref{tab:sim_n300} strongly reinforce the robustness advantage of DAR-R. It achieves $\mathrm{TP}=10.00$ in all settings and keeps FP low (about $1.3$--$2.1$ across all $n$ and scenarios). By contrast, Penalized-LME and Robust-LME both miss true signals and over-select noise (typically TP $\approx 7.6$--$9.5$ and FP $\approx 15$--$16$), while LASSO also attains $\mathrm{TP}=10.00$ but with extremely large FP (about $31$--$34$). Thus, DAR-R combines perfect sensitivity with strong specificity, and this remains stable under both S2 and S3. A similarly clear pattern appears for covariance estimation. The Frobenius error $\|\hat D-D\|_F$ is smallest for DAR-R in every setting, and the gap is substantial. For instance, DAR-R is about $118$--$120$ at $n=100$ and improves to about $56$--$62$ at $n=300$ (under the $10^3$ scaling), whereas the competing mixed-model methods remain much larger (roughly $270$--$356$). Figure~\ref{fig:cov_error_boxplot} supports this numerically, showing that DAR-R not only attains the lowest covariance error but also has the tightest boxplots. This indicates that the doubly adaptive weighting stabilizes both fixed-effect estimation and the random-effects covariance update.

The MSPE results in Tables~\ref{tab:sim_mse100_combined}--\ref{tab:sim_mse300_combined} and Figure~\ref{fig:mspe_boxplot} require a more nuanced interpretation. Because prediction is evaluated on a separately generated \emph{clean} test set, LASSO often attains the lowest MSPE, while DAR-R is typically close to Oracle-LME and substantially better than Penalized-LME and Robust-LME. This does not contradict the main goal of the paper: LASSO achieves lower MSPE at the cost of severe over-selection (FP $\approx 32$--$34$) and no meaningful covariance recovery, whereas DAR-R maintains competitive prediction while delivering much stronger structural recovery and longitudinal covariance estimation.

Finally, the sample-size trend is stable and aligns with theory. As $n$ increases from 100 to 300, DAR-R improves monotonically in $\mathrm{MSE}(S)$, $\mathrm{MSE}(S^c)$, and covariance error, and its boxplots become more concentrated. The relative ordering of methods remains essentially unchanged across $n$ and across contamination scenarios, which provides strong empirical support for the robustness and consistency results in Section~\ref{sec:theory}. Overall, the six tables and four boxplots show that DAR-R achieves the best balance of robustness, sparsity recovery, interpretability, and covariance learning in high-dimensional longitudinal mixed-effects models.

\begin{table}[h]
\centering
\caption{Simulation Results ($n=100$): Estimation Error (multiplied $10^3$).}
\label{tab:sim_mse100_combined}
\resizebox{\textwidth}{!}{%
\begin{tabular}{l|ccc|ccc|ccc}
\toprule
\multirow{2}*{Method} & \multicolumn{3}{c|}{S1: Clean} & \multicolumn{3}{c|}{S2: Vertical Outliers} & \multicolumn{3}{c}{S3: Bad Leverage} \\
\cline{2-10} \addlinespace
& MSE(S) & MSE($S^c$) & MSPE & MSE(S) & MSE($S^c$) & MSPE & MSE(S) & MSE($S^c$) & MSPE \\ \midrule
DAR-R & 4.0014 & 0.0175 & 840.1898 & 4.0222 & 0.0128 & 840.0912 & 4.1893 & 0.0239 & 841.9487 \\
 & (3.7584) & (0.0000) & (837.5673) & (3.6638) & (0.0000) & (837.6220) & (3.8741) & (0.0000) & (839.9249) \\
Penalized-LME & 438.6862 & 1.6500 & 2561.3839 & 435.7005 & 1.6430 & 2537.4974 & 430.7705 & 1.5535 & 2538.2281 \\
 & (446.7335) & (1.5128) & (2641.2302) & (437.7787) & (1.4372) & (2544.0819) & (453.4023) & (1.4815) & (2510.3144) \\
Robust-LME & 276.8216 & 1.3124 & 1996.1087 & 282.6346 & 1.3244 & 2041.7149 & 274.1593 & 1.2668 & 2002.8412 \\
 & (296.3997) & (1.1851) & (2054.0077) & (304.4818) & (1.1456) & (2050.5558) & (281.7313) & (1.1766) & (2077.0142) \\
LASSO & 32.6688 & 0.7961 & 836.7042 & 34.6796 & 0.7475 & 835.2566 & 31.9727 & 0.7896 & 830.6293 \\
 & (30.9947) & (0.6866) & (829.7465) & (33.7228) & (0.6573) & (836.2839) & (30.6092) & (0.7143) & (830.5921) \\
Oracle-LME & 4.1594 & 0.1891 & 895.6626 & 4.1509 & 0.1914 & 895.9017 & 4.3045 & 0.1924 & 896.4514 \\
 & (3.8762) & (0.1672) & (889.9236) & (3.6481) & (0.1677) & (889.1046) & (3.8583) & (0.1684) & (896.5189) \\ \bottomrule
\end{tabular}}
\end{table}

\begin{table}[!]
\centering
\caption{Simulation Results ($n=200$): Estimation Error (multiplied $10^3$).}
\label{tab:sim_mse200_combined}
\resizebox{\textwidth}{!}{%
\begin{tabular}{l|ccc|ccc|ccc}
\toprule
\multirow{2}*{Method} & \multicolumn{3}{c|}{S1: Clean} & \multicolumn{3}{c|}{S2: Vertical Outliers} & \multicolumn{3}{c}{S3: Bad Leverage} \\
\cline{2-10} \addlinespace
& MSE(S) & MSE($S^c$) & MSPE & MSE(S) & MSE($S^c$) & MSPE & MSE(S) & MSE($S^c$) & MSPE \\ \midrule
DAR-R & 2.0141 & 0.0089 & 830.7269 & 1.9291 & 0.0056 & 823.7188 & 2.0034 & 0.0105 & 816.7024 \\
 & (1.8200) & (0.0000) & (831.0803) & (1.8057) & (0.0000) & (825.1815) & (1.8849) & (0.0000) & (817.6517) \\
Penalized-LME & 354.7264 & 0.6495 & 2181.6711 & 345.7907 & 0.6865 & 2160.6638 & 348.6577 & 0.6197 & 2168.6971 \\
 & (328.8987) & (0.5762) & (2126.8027) & (326.1388) & (0.6164) & (2113.6322) & (330.5244) & (0.5397) & (2110.1389) \\
Robust-LME & 174.0144 & 0.4714 & 1537.4667 & 167.0265 & 0.4853 & 1501.8733 & 168.1292 & 0.4550 & 1502.5344 \\
 & (138.3350) & (0.3958) & (1409.0990) & (135.8225) & (0.4252) & (1380.0560) & (137.8146) & (0.3845) & (1390.8322) \\
LASSO & 16.3756 & 0.3587 & 752.1113 & 15.8671 & 0.3409 & 745.3660 & 16.4342 & 0.3713 & 742.8990 \\
 & (15.1770) & (0.3262) & (750.8083) & (14.4914) & (0.3062) & (738.7506) & (15.5637) & (0.3434) & (739.6295) \\
ORACLE & 2.0811 & 0.0945 & 876.1544 & 1.9552 & 0.0911 & 869.0523 & 2.0455 & 0.0907 & 862.0626 \\
 & (1.8718) & (0.0821) & (875.7225) & (1.7789) & (0.0794) & (869.5958) & (1.9169) & (0.0840) & (859.8074) \\ \bottomrule
\end{tabular}}
\end{table}

\begin{table}[!]
\centering
\caption{Simulation Results ($n=300$): Estimation Error (multiplied $10^3$).}
\label{tab:sim_mse300_combined}
\resizebox{\textwidth}{!}{%
\begin{tabular}{l|ccc|ccc|ccc}
\toprule
\multirow{2}*{Method} & \multicolumn{3}{c|}{S1: Clean} & \multicolumn{3}{c|}{S2: Vertical Outliers} & \multicolumn{3}{c}{S3: Bad Leverage} \\
\cline{2-10} \addlinespace
& MSE(S) & MSE($S^c$) & MSPE & MSE(S) & MSE($S^c$) & MSPE & MSE(S) & MSE($S^c$) & MSPE \\ \midrule
DAR-R & 1.3250 & 0.0047 & 830.6165 & 1.2717 & 0.0067 & 816.7730 & 1.3601 & 0.0074 & 828.9421 \\
 & (1.2181) & (0.0000) & (829.2662) & (1.1283) & (0.0000) & (816.0352) & (1.2471) & (0.0000) & (824.7018) \\
Penalized-LME & 274.3064 & 0.3715 & 1885.0329 & 282.8742 & 0.3590 & 1902.8627 & 281.5650 & 0.3904 & 1913.6558 \\
 & (315.0260) & (0.3274) & (1889.0135) & (314.0539) & (0.3313) & (1845.2760) & (315.1112) & (0.3488) & (1887.9920) \\
Robust-LME & 89.8391 & 0.2437 & 1211.0292 & 102.2153 & 0.2323 & 1252.6106 & 100.1315 & 0.2608 & 1252.0943 \\
 & (128.0551) & (0.2093) & (1227.1709) & (128.7315) & (0.2200) & (1264.7300) & (127.9165) & (0.2268) & (1247.3950) \\
LASSO & 10.0129 & 0.2556 & 730.6362 & 10.3956 & 0.2253 & 718.1238 & 10.5066 & 0.2377 & 725.8174 \\
 & (9.4367) & (0.2379) & (728.4128) & (10.1073) & (0.2065) & (718.8750) & (9.9150) & (0.2062) & (721.5650) \\
ORACLE & 1.3588 & 0.0596 & 873.2835 & 1.2836 & 0.0557 & 858.5259 & 1.3747 & 0.0600 & 871.9683 \\
 & (1.2043) & (0.0567) & (874.7262) & (1.1273) & (0.0487) & (858.0622) & (1.2547) & (0.0515) & (866.1297) \\ \bottomrule
\end{tabular}}
\end{table}

\begin{table}[!]
\centering
\caption{Simulation Results ($n=100$): Selection (TP/FP) and Covariance Error (multiplied $10^3$).}
\label{tab:sim_n100}
\scalebox{0.75}{%
\begin{tabular}{@{}l|cc|cc|cc@{}}
\toprule
\multirow{2}*{Method} & \multicolumn{2}{c|}{S1: Clean} & \multicolumn{2}{c|}{S2: Vertical Outliers} & \multicolumn{2}{c}{S3: Bad Leverage} \\
\cline{2-7} \addlinespace
& TP / FP & $\|\hat{D}-D\|_F$ & TP / FP & $\|\hat{D}-D\|_F$ & TP / FP & $\|\hat{D}-D\|_F$ \\ \midrule
DAR-R & 10.00 / 1.58 & 117.59 (101.18) & 10.00 / 1.31 & 118.08 (98.12) & 10.00 / 2.13 & 119.67 (107.62) \\
Penalized-LME & 7.60 / 16.43 & 317.08 (304.35) & 7.68 / 16.38 & 314.51 (299.07) & 7.66 / 16.38 & 322.69 (308.48) \\
Robust-LME & 8.46 / 16.55 & 336.77 (297.42) & 8.50 / 16.51 & 356.33 (304.50) & 8.48 / 16.52 & 339.91 (295.79) \\
LASSO & 10.00 / 33.75 & - & 10.00 / 32.69 & - & 10.00 / 33.81 & - \\
Oracle-LME & 10.00 / 6.00 & 295.42 (270.65) & 10.00 / 6.00 & 294.67 (272.18) & 10.00 / 6.00 & 299.59 (271.04) \\ \bottomrule
\end{tabular}}
\end{table}

\begin{table}[!]
\centering
\caption{Simulation Results ($n=200$): Selection (TP/FP) and Covariance Error (multiplied $10^3$).}
\label{tab:sim_n200}
\scalebox{0.75}{%
\begin{tabular}{@{}l|cc|cc|cc@{}}
\toprule
\multirow{2}*{Method} & \multicolumn{2}{c|}{S1: Clean} & \multicolumn{2}{c|}{S2: Vertical Outliers} & \multicolumn{2}{c}{S3: Bad Leverage} \\
\cline{2-7} \addlinespace
& TP / FP & $\|\hat{D}-D\|_F$ & TP / FP & $\|\hat{D}-D\|_F$ & TP / FP & $\|\hat{D}-D\|_F$ \\ \midrule
DAR-R & 10.00 / 1.41 & 84.04 (71.64) & 10.00 / 1.32 & 73.11 (64.81) & 10.00 / 1.85 & 84.38 (73.57) \\
Penalized-LME & 8.06 / 15.95 & 307.88 (300.23) & 8.09 / 15.91 & 300.07 (290.46) & 8.10 / 15.90 & 299.43 (288.70) \\
Robust-LME & 9.01 / 15.99 & 290.22 (270.37) & 9.05 / 15.96 & 302.23 (276.26) & 9.06 / 15.95 & 285.94 (263.93) \\
LASSO & 10.00 / 32.84 & - & 10.00 / 32.24 & - & 10.00 / 33.14 & - \\
Oracle-LME & 10.00 / 6.00 & 284.96 (267.04) & 10.00 / 6.00 & 280.82 (263.25) & 10.00 / 6.00 & 277.45 (262.09) \\ \bottomrule
\end{tabular}}
\end{table}

\begin{table}[!]
\centering
\caption{Simulation Results ($n=300$): Selection (TP/FP) and Covariance Error (multiplied $10^3$).}
\label{tab:sim_n300}
\scalebox{0.75}{%
\begin{tabular}{@{}l|cc|cc|cc@{}}
\toprule
\multirow{2}*{Method} & \multicolumn{2}{c|}{S1: Clean} & \multicolumn{2}{c|}{S2: Vertical Outliers} & \multicolumn{2}{c}{S3: Bad Leverage} \\
\cline{2-7} \addlinespace
& TP / FP & $\|\hat{D}-D\|_F$ & TP / FP & $\|\hat{D}-D\|_F$ & TP / FP & $\|\hat{D}-D\|_F$ \\ \midrule
DAR-R & 10.00 / 1.41 & 55.65 (49.96) & 10.00 / 1.74 & 62.37 (51.13) & 10.00 / 1.87 & 60.20 (50.33) \\
Penalized-LME & 8.47 / 15.54 & 287.13 (280.96) & 8.45 / 15.56 & 288.64 (282.29) & 8.42 / 15.58 & 291.50 (283.23) \\
Robust-LME & 9.45 / 15.55 & 272.89 (256.37) & 9.43 / 15.57 & 276.66 (262.02) & 9.41 / 15.60 & 279.51 (261.04) \\
LASSO  & 10.00 / 33.36 & - & 10.00 / 31.61 & - & 10.00 / 32.92 & - \\
Oracle-LME & 10.00 / 6.00 & 271.13 (260.25) & 10.00 / 6.00 & 272.35 (259.04) & 10.00 / 6.00 & 277.41 (262.68) \\ \bottomrule
\end{tabular}}
\end{table}

\begin{figure}[!]
\centering
\makebox[\textwidth][c]{\includegraphics[width=0.95\textwidth]{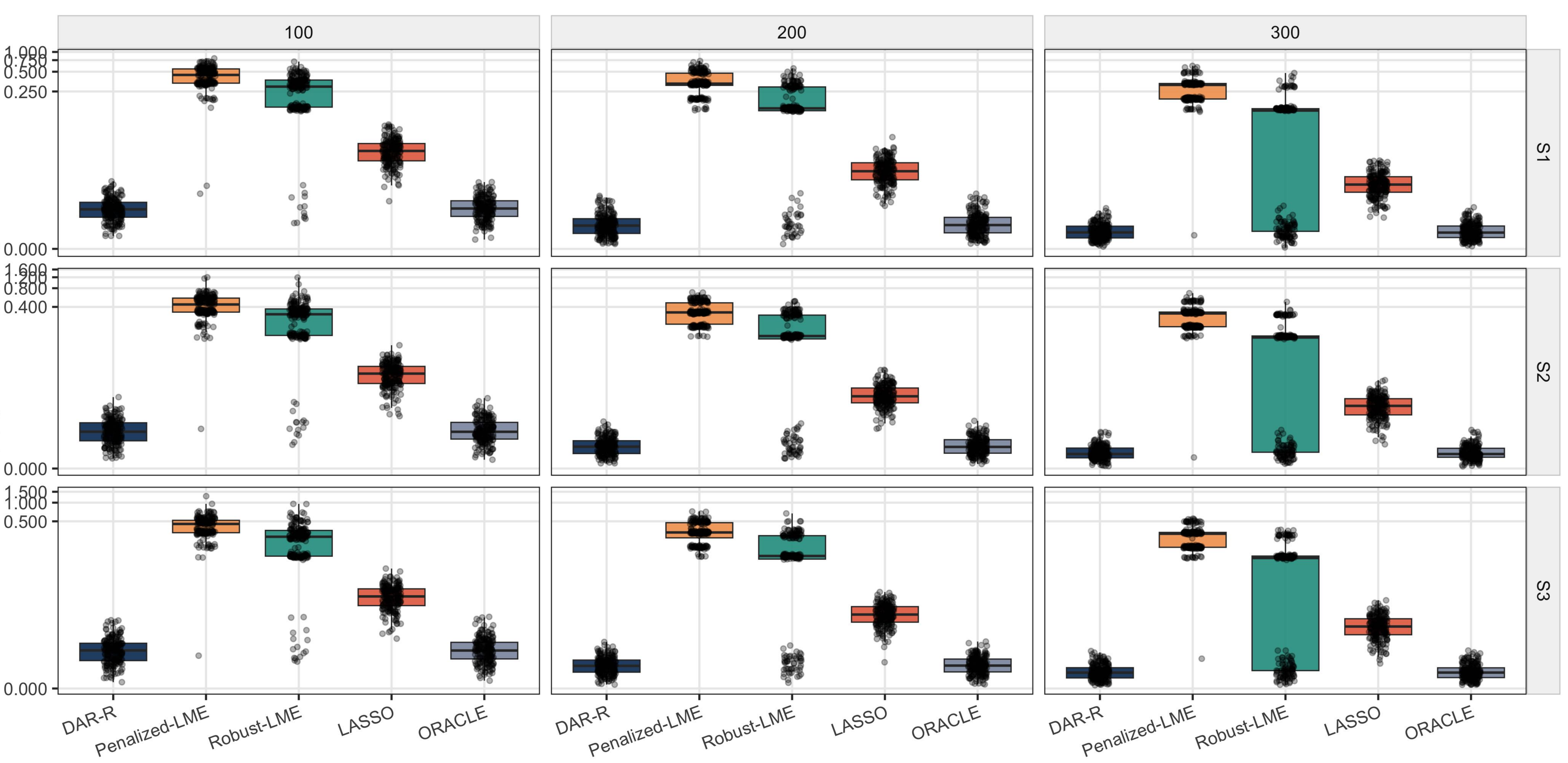}}
\caption{Boxplots of $\mathrm{MSE}(S)$ under scenarios S1--S3. All values are $\log_{10}$-transformed).}
\label{fig:mseS_boxplot}
\end{figure}

\begin{figure}[h]
\centering
\makebox[\textwidth][c]{\includegraphics[width=0.95\textwidth]{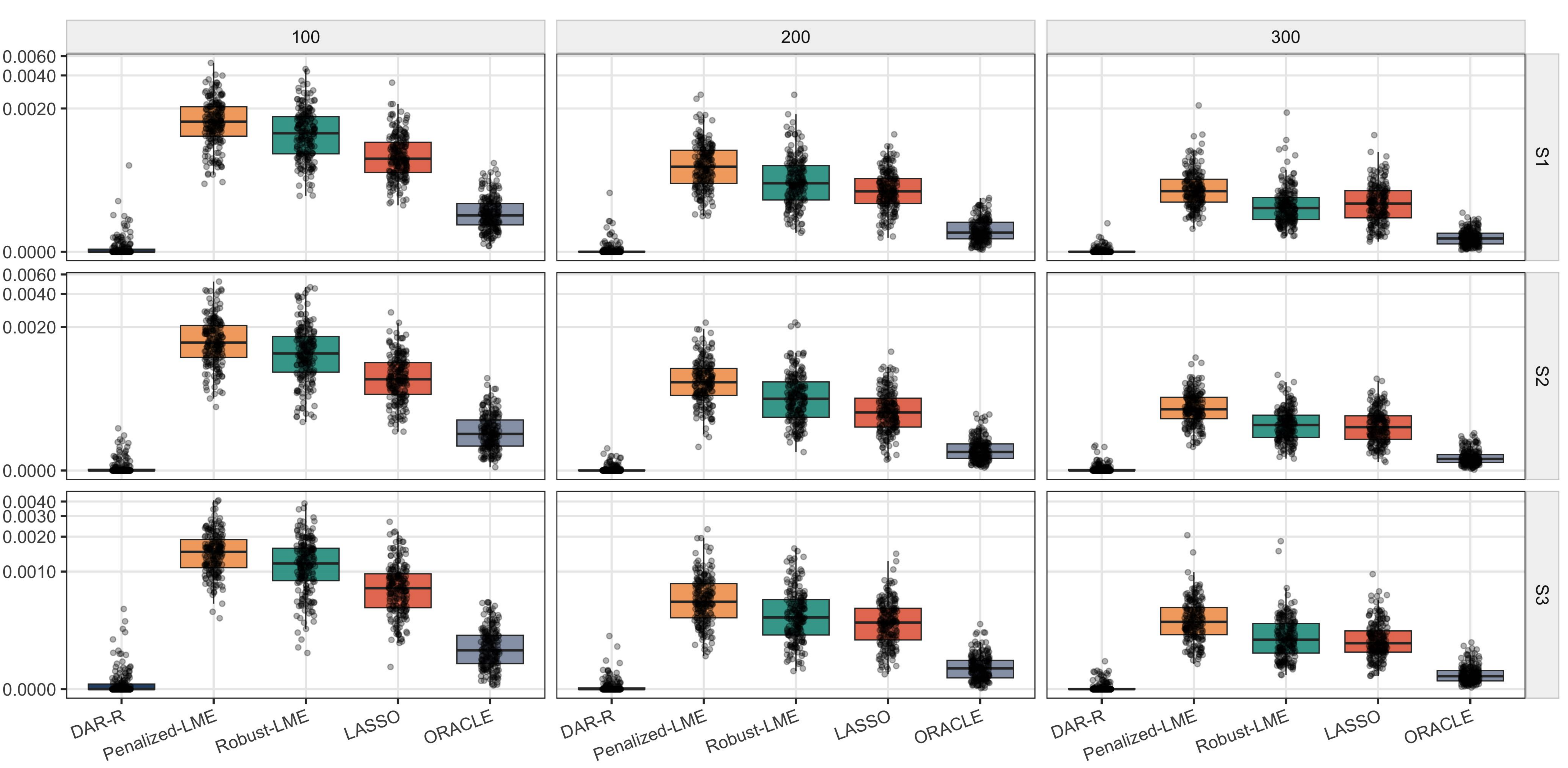}}
\caption{Boxplots of $\mathrm{MSE}(S^c)$ under scenarios S1--S3. All values are $\log_{10}$-transformed.}
\label{fig:mseSc_boxplot}
\end{figure}

\begin{figure}[!]
\centering
\makebox[\textwidth][c]{\includegraphics[width=0.95\textwidth]{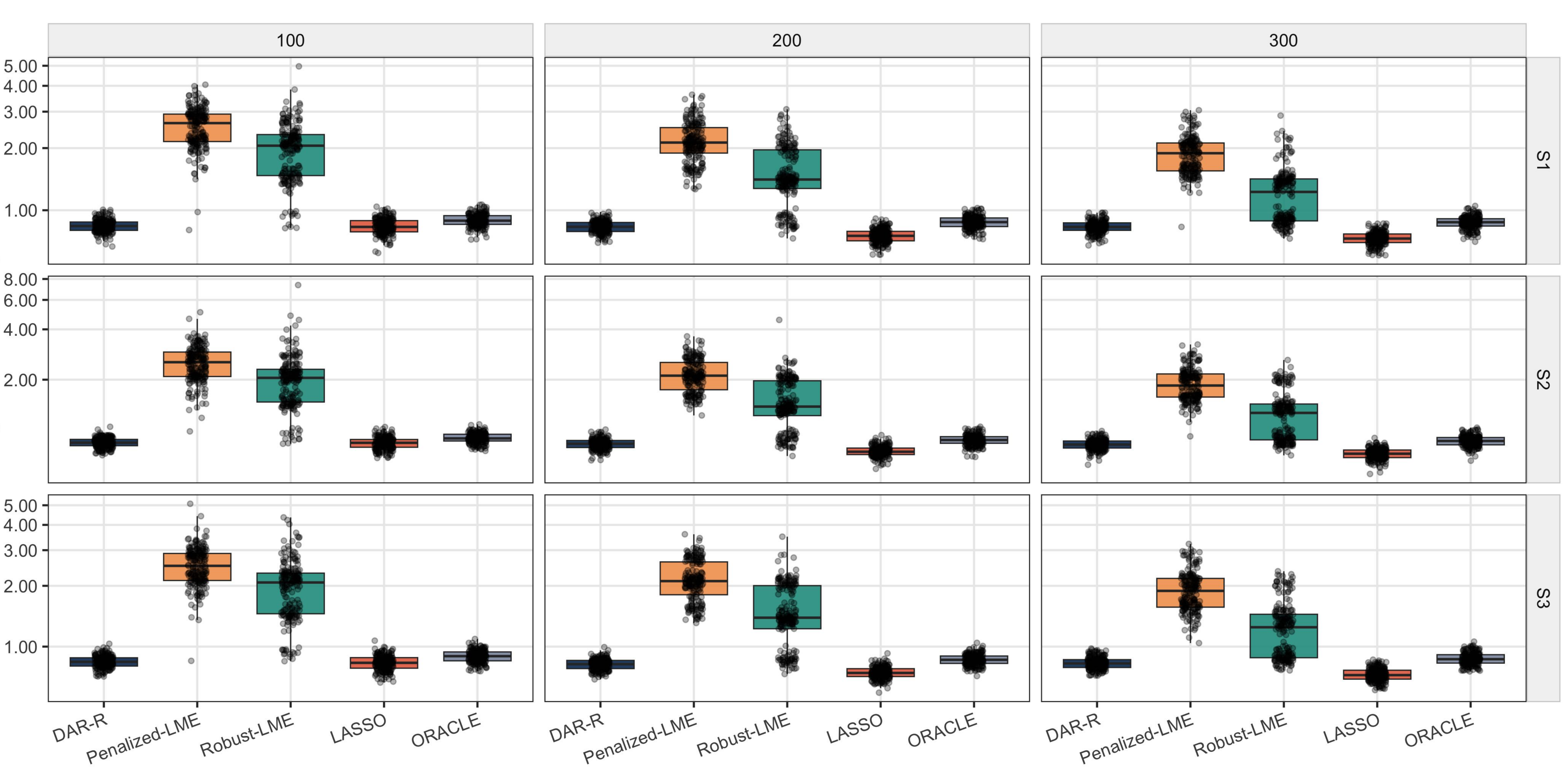}}
\caption{Boxplots of MSPE under scenarios S1--S3. All values are $\log_{10}$-transformed.}
\label{fig:mspe_boxplot}
\end{figure}

\begin{figure}[h]
\centering
\makebox[\textwidth][c]{\includegraphics[width=0.95\textwidth]{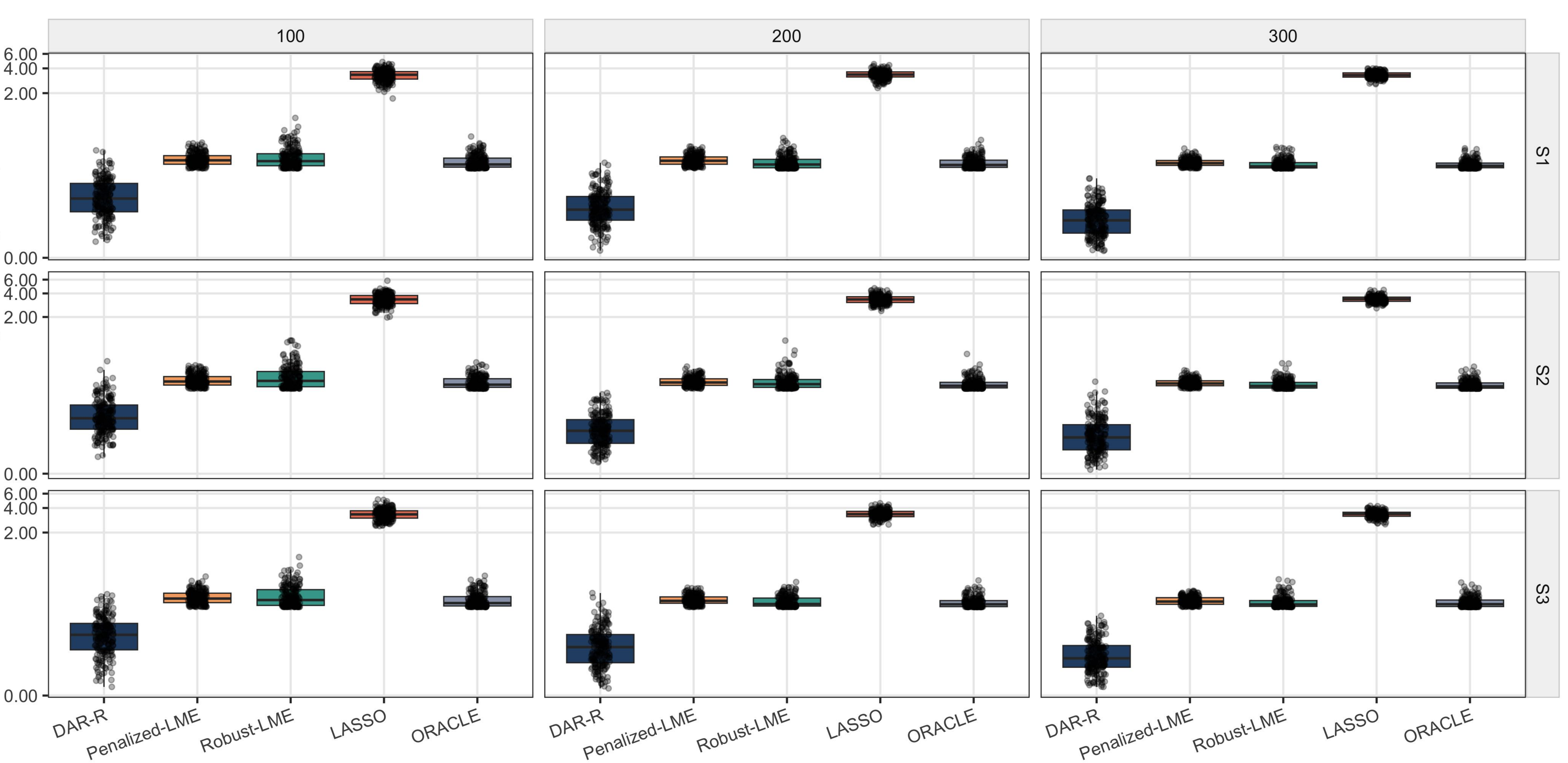}}
\caption{Boxplots of covariance estimation error for the random-effect. All values are $\log_{10}$-transformed.}
\label{fig:cov_error_boxplot}
\end{figure}

%%%%%%%%%%%%%%%%%%%%%%%%%%%%%%%%%%%%%%%%%%%%%%%%%%%%%%%%%%%%%%%%%%%%%%%%
\section{Application Study}\label{sec:application}
\subsection{Data Description}

We illustrate the proposed DAR-R method using the TADPOLE challenge data\footnote{\texttt{https://tadpole.grand-challenge.org/Data/}}, which is derived from the Alzheimer's Disease Neuroimaging Initiative (ADNI) and was designed for forecasting longitudinal Alzheimer's disease progression. In the TADPOLE framework, each row of the main table corresponds to one subject visit and each column corresponds to a feature/measurement, making it a natural testbed for high-dimensional longitudinal modeling. The challenge design also explicitly distinguishes longitudinal training and prediction subsets (D1 and D2), which is well aligned with our objective of robust estimation and out-of-sample validation in a repeated-measures setting.

Our primary analysis uses the file \texttt{TADPOLE\_D1\_D2.csv}.The data contain repeated observations for the same subject indexed by \texttt{RID}, with visit identifiers such as \texttt{VISCODE} and examination dates \texttt{EXAMDATE}. This yields a genuinely longitudinal structure suitable for mixed-effects modeling. The file also contains a large number of candidate predictors (demographics, cognitive scores, MRI-derived measurements, PET/biomarker variables, genetics, and other clinical measurements), which is appropriate for the high-dimensional sparse fixed-effect component in DAR-R.

We focus on a continuous longitudinal outcome and take $Y_{it}=\texttt{ADAS13}_{it}$, the ADAS-Cog 13 score for subject $i$ at visit $t$, as the primary response. This choice is clinically meaningful and also consistent with the TADPOLE challenge, where ADAS-Cog 13 is one of the key forecast targets. As a sensitivity analysis, the same pipeline can be repeated with \texttt{MMSE} or \texttt{CDRSB} as the outcome.

The time variable is constructed from \texttt{EXAMDATE}. For each subject, we define
\[
t_{it}=\frac{\text{EXAMDATE}_{it}-\text{EXAMDATE}_{i,\mathrm{baseline}}}{30.44},
\]
i.e., months from the subject-specific baseline visit. This creates a continuous follow-up time scale and allows a random-slope specification. We use the random-effects design
\[
Z_{it}=(1,\; t_{it})^\top,
\]
so that the model includes subject-specific random intercepts and random time slopes.

For the fixed effects $X_{it}$, we include a clinically interpretable core set (e.g., age, sex, education, APOE4) together with a high-dimensional block of imaging/biomarker variables available in \texttt{TADPOLE\_D1\_D2.csv}. In practice, we recommend starting from MRI-derived volumetric features (e.g., hippocampus, entorhinal cortex, ventricles, whole brain) and then expanding to additional modalities if missingness is manageable. Because the TADPOLE table is very wide, we apply a prespecified screening and preprocessing pipeline (described below) before fitting DAR-R.

To mimic the original TADPOLE longitudinal prediction setup, we exploit the dataset indicators \texttt{D1} and \texttt{D2}. In particular, D1 is used as the primary development/training sample, while D2 is used as an external longitudinal evaluation set whenever possible. This split is attractive because it respects the intended challenge design and reduces the risk of optimistic in-sample assessment.

\subsection{Implementation}

We first preprocess the TADPOLE data as follows. We retain visits with non-missing \texttt{RID}, \texttt{EXAMDATE}, and the chosen outcome (\texttt{ADAS13} in the main analysis). We remove duplicate subject-visit entries if present, sort observations by subject and date, and construct $t_{it}$ (months from baseline). Continuous predictors are standardized using robust location/scale estimates computed on the training set only, and categorical variables are dummy-coded. We exclude predictors with extremely high missingness (e.g., missing rate above a prespecified threshold such as 60\%--70\%), near-zero variance, or obvious duplication of the same measurement under multiple aliases. Remaining missing values are imputed within the training folds only (e.g., robust median for continuous variables and mode for categorical variables), with the same transformations carried to validation/test data to avoid leakage. 

Our primary fitted model is the proposed DAR-R longitudinal mixed-effects model with outcome \texttt{ADAS13}, random intercept and random slope in time, and a high-dimensional fixed-effect vector penalized by MCP (or SCAD). The method uses a robust pilot fit, doubly adaptive observation weights (residual-based and leverage-based), and weighted random-effects/variance updates as developed in Sections~\ref{sec:method}--\ref{sec:theory}. Tuning parameters are selected on the D1 sample using subject-level cross-validation (i.e., all visits from a subject are kept in the same fold) to preserve the longitudinal dependence structure. In the application study, we do not generate synthetic datasets. 

We use the same TADPOLE data and run 100 repeated subject-level holdout experiments (R = 100, $\text{D1}_\text{CV}$). In each repeat, subjects are randomly split into 75\% training and 25\% testing sets, each method is refit on the training set, and performance is evaluated on the held-out subjects. Missing values are handled within each training fold, predictors are robustly standardized, and \texttt{time-month} is derived from VISCODE/EXAMDATE. All reported results are summarized as mean (SD) over the 100 repeats, and the predictor table lists the data-driven top 10 variables by selection frequency.

We compare DAR-R against several baselines chosen to reflect different modeling priorities:
\begin{itemize}
    \item \textbf{Penalized-LME}: a non-robust penalized linear mixed-effects model with the same random-effects structure.
    \item \textbf{Robust-LME}: a robust mixed-effects model (without high-dimensional sparse selection, or with a reduced feature set when needed).
    \item \textbf{LASSO}: a sparse regression model that ignores subject-level random effects, included as a prediction-oriented benchmark.
    \item \textbf{Mixed-effects Model}: a low-dimensional mixed-effects model using core covariates (e.g., age, sex, education, APOE4, and time), motivated by the TADPOLE benchmark mixed-effects strategy.
\end{itemize}
This comparison set is deliberate: it separates the contributions of robustness, sparsity, and mixed-effects modeling.

For reproducibility, all preprocessing steps, feature screening thresholds, penalty grids, and cross-validation splits are fixed before model fitting. In addition, we repeat the subject-level cross-validation multiple times (or use repeated bootstrap resampling at the subject level) to assess the stability of variable selection and predictive performance.

\subsection{Evaluation Criteria}

Since the true regression coefficients and random-effects covariance are unknown in real data, we replace simulation-based accuracy metrics (e.g., $\|\hat\beta-\beta^\star\|$ or $\|\hat D-D\|$) with predictive, stability, and model-diagnostics criteria that are appropriate for longitudinal observational data.

We use two complementary evaluation settings. The first is internal validation on D1 via subject-level cross-validation for tuning and model comparison. The second is external longitudinal validation using D2 (when the required variables are available), where the model is trained on D1 and evaluated on D2. This D1$\rightarrow$D2 design is particularly useful because it follows the intended TADPOLE split and better reflects out-of-sample generalization.

For continuous longitudinal prediction of \texttt{ADAS13}, we report:
\begin{itemize}
    \item \textbf{Mean Absolute Error(MAE)}:
    \[
    \mathrm{MAE}
    =
    \frac{1}{N_{\mathrm{test}}}
    \sum_{i,t}\big|Y_{it}-\hat Y_{it}\big|,
    \]
    which is easy to interpret clinically and is less sensitive to extreme errors.
    
    \item \textbf{Root Mean Sqaured Error(RMSE)}:
    \[
    \mathrm{RMSE}
    =
    \left(
    \frac{1}{N_{\mathrm{test}}}
    \sum_{i,t}(Y_{it}-\hat Y_{it})^2
    \right)^{1/2},
    \]
    which emphasizes larger prediction errors and complements MAE.
    
    \item \textbf{Median Absolute Error (MedAE)}, a robust summary of prediction error that is especially informative in the presence of atypical visits.
\end{itemize}

Because longitudinal studies often contain irregular visits and subject heterogeneity, we also report averaged errors: we first compute each subject's mean absolute prediction error across visits and then summarize these subject-specific errors (median and IQR). This prevents subjects with many visits from dominating the evaluation.

To assess variable selection behavior in the absence of ground truth, we evaluate:
\begin{itemize}
    \item \textbf{Model size}: the number of selected fixed effects (excluding mandatory nuisance covariates such as time and intercept terms).
    \item \textbf{Selection stability}: the average Jaccard similarity of selected sets across repeated subject-level resamples/folds,
    \[
    \mathrm{Jaccard}(A,B)=\frac{|A\cap B|}{|A\cup B|}.
    \]
    A robust and well-calibrated sparse model should not only predict well but also select a reasonably stable set of predictors.
    \item \textbf{Sign consistency}: for repeatedly selected predictors, the proportion of resamples in which the estimated sign remains unchanged.
\end{itemize}

%To evaluate the mixed-effects component and robustness mechanism, we report:
%\begin{itemize}
%    \item \textbf{Conditional residual diagnostics}: distribution of weighted residuals and subject-level residual summaries over time.
    %\item \textbf{Downweighting diagnostics}: proportion of observations with weights below prespecified thresholds (e.g., $w_{it}<0.5$), and inspection of whether these visits correspond to clinically implausible jumps or high-leverage covariate patterns.
%    \item \textbf{Random-effects covariance diagnostics}: estimated variance components and correlation (random intercept--slope correlation), together with sensitivity checks across tuning choices.
%\end{itemize}

Finally, because prediction and interpretability must be balanced in this application, we summarize all methods using a joint perspective: predictive accuracy (MAE/RMSE/MedAE), sparsity/stability (model size + Jaccard), and longitudinal interpretability (whether selected predictors and random-effects patterns are clinically plausible).

\subsection{Results}
We present the application results in four parts: predictive performance, variable selection patterns, robustness diagnostics, and longitudinal interpretation.

\begin{table}[h]
\centering
\caption{Application results on TADPOLE.}
\label{tab:app_main_results}
\scalebox{0.9}{
\begin{tabular}{lrrrrrr}
\toprule
Method & MAE & RMSE & MedAE & Model size & Jaccard stability  \\
\midrule
DAR-R (MCP)      &  \texttt{5.76 (0.34)} & \texttt{7.20 (0.39)} & \texttt{4.97 (0.52)} & \texttt{12.1 (2.3)} & \texttt{0.720 (0.114)}  \\
DAR-R (SCAD)     &  \texttt{5.76 (0.34)} & \texttt{7.20 (0.39)} & \texttt{4.97 (0.52)} & \texttt{13.6 (2.4)} & \texttt{0.720 (0.096)}  \\
Penalized-LME    &  \texttt{14.87 (0.23)} & \texttt{14.97 (0.23)} & \texttt{14.78 (0.23)} & \texttt{20.0 (0.0)} & \texttt{0.986 (0.034)}  \\
Robust-LME       &  \texttt{14.87 (0.23)} & \texttt{14.97 (0.23)} & \texttt{14.74 (0.23)} &  \texttt{low-dim} & --  \\
LASSO&  \texttt{6.07 (0.35)} & \texttt{7.58 (0.42)} & \texttt{5.24 (0.55)} & \texttt{17.8 (1.9)} & \texttt{0.810 (0.084)}  \\
LME    &  \texttt{16.22 (0.93)} & \texttt{19.55 (0.96)} & \texttt{13.95 (1.13)} & \texttt{fixed} & -- \\
\bottomrule
\end{tabular}}
\end{table}

\begin{table}[h]
\centering
\caption{Top predictors selected by DAR-R across repeated subject-level resamples. }
\label{tab:app_stability_predictors}
\scalebox{0.95}{
\begin{tabular}{llccrc}
\toprule
Rank & Predictor & Selection\_fr & Sign\_consis & Median\_coef & IQR\_coefficient \\
\midrule
1  & \texttt{ADAS11}                  & \texttt{1}    & \texttt{1} & \texttt{6.543} & \texttt{0.130} \\
2  & \texttt{APOE4}                    & \texttt{1}    & \texttt{1} & \texttt{5.876} & \texttt{0.276} \\
3  & \texttt{PTGENDER}                  & \texttt{1}    & \texttt{1} & \texttt{-5.242} & \texttt{0.223} \\
4  & \texttt{AGE}                       & \texttt{1}    & \texttt{1} & \texttt{2.204} & \texttt{0.453} \\
5  & \texttt{DX}                        & \texttt{1}    & \texttt{1} & \texttt{0.842} & \texttt{0.494} \\
6  & \texttt{RAVLT-immediate}            & \texttt{0.99} & \texttt{1} & \texttt{-0.266} & \texttt{0.099} \\
7  & \texttt{time-month}                 & \texttt{0.97} & \texttt{1} & \texttt{0.202} & \texttt{0.044} \\
8  & \texttt{RAVLT-learning}             & \texttt{0.91} & \texttt{1} & \texttt{-0.112} & \texttt{0.083} \\
9  & \texttt{Ventricles}                 & \texttt{0.91} & \texttt{1} & \texttt{-0.059} & \texttt{0.055} \\
10 & \texttt{CDRSB}                      & \texttt{0.86} & \texttt{1} & \texttt{0.107} & \texttt{0.080} \\
\bottomrule
\end{tabular}}
\end{table}

\begin{figure}[!]
\centering
\includegraphics[width=0.7\textwidth]{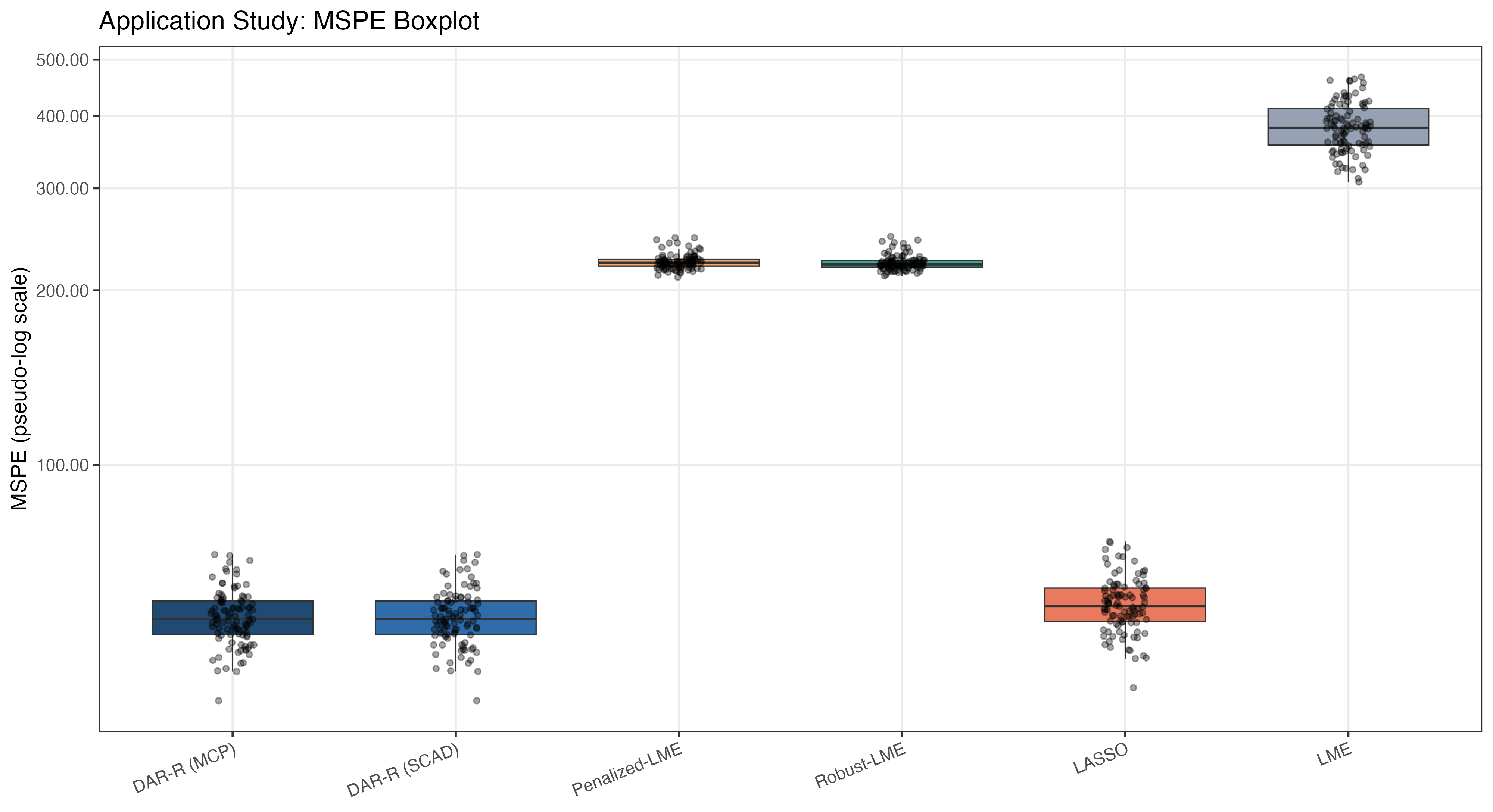}
\caption{Boxplots of the application-study prediction error (MSPE) across repeated  resampling runs for all competing methods on the TADPOLE data.}
\label{fig:app_mspe_boxplot}
\end{figure}

Among all competing methods, DAR-R (under either the MCP or SCAD penalty) delivers the strongest predictive performance for ADAS13, achieving an MAE of 5.76 under both penalties and outperforming all comparators by a clear margin. In \hyperref[tab:app_main_results]{Table 7}, DAR-R attains the smallest MAE (5.76), RMSE (7.20), and MedAE (4.97), with small variability across repeated resamples, indicating stable and reproducible prediction accuracy. Although LASSO also performs reasonably well in prediction (MAE 6.07), it does so with a noticeably larger model (average size 17.8), compared with the much sparser DAR-R model (about 12--13 predictors), suggesting over-selection when subject-level longitudinal dependence is ignored. By contrast, Penalized-LME and Robust-LME both yield substantially larger prediction errors (MAE $\approx 14.87$), showing that standard penalized mixed models and conventional robust mixed-effects procedures are not sufficient in this high-dimensional longitudinal setting. The ability of DAR-R to combine the best prediction accuracy with a compact model strongly supports the practical value of the proposed doubly adaptive weighting strategy in real clinical data.

Beyond prediction, the real-data analysis also demonstrates the interpretability and selection stability of DAR-R. The top predictors repeatedly selected across subject-level resamples, including ADAS11, APOE4, PTGENDER, AGE, and DX, are clinically meaningful and consistent with established knowledge on Alzheimer's disease progression, which strengthens confidence in the method's scientific validity. In particular, ADAS11 (median coefficient 6.543) emerges as the dominant predictor of ADAS13, consistent with the clinical expectation that baseline cognitive impairment is a major determinant of subsequent decline. The strong and stable contribution of APOE4 (median coefficient 5.876) is also medically plausible, reflecting its central role as a genetic risk factor for neurodegenerative progression. In addition, the inclusion of Ventricles, CDRSB, and RAVLT-related measures suggests that DAR-R captures a clinically coherent signal spanning structural neurodegeneration, global dementia severity, and memory dysfunction. Importantly, all top predictors have sign consistency equal to 1, indicating that DAR-R not only selects relevant variables but also recovers their directional effects reliably across resamples, which is essential for clinically interpretable inference. The Jaccard summary further shows that DAR-R achieves a strong stability--sparsity balance: its average Jaccard coefficient (0.72) is lower than that of the less selective Penalized-LME (0.986), but remains highly competitive given the substantially smaller model size, while LASSO's slightly higher stability (0.81) is obtained with a much larger and less parsimonious model.

The resampling boxplots in \hyperref[fig:app_mspe_boxplot]{Figure 5} provide an additional robustness perspective. DAR-R not only has the lowest median prediction error, but also shows the tightest spread, indicating limited sensitivity to resampling variation and improved robustness to heterogeneous patient profiles. This empirical stability is consistent with the design of DAR-R: the doubly adaptive weighting mechanism simultaneously downweights response outliers and high-leverage covariate observations, thereby reducing the influence of implausible visits, recording anomalies, and atypical biomarker patterns that frequently arise in real-world longitudinal cohorts. From a biomedical standpoint, this is particularly important because ADNI/TADPOLE data are collected across multiple sites and visits, where variation in assessment conditions, imaging quality, and patient compliance can introduce substantial noise. A method that remains stable under such conditions is more suitable for longitudinal risk monitoring and decision-support applications, where robustness and interpretability are as important as raw predictive accuracy.

Taken together, the TADPOLE application results are fully consistent with the simulation findings: DAR-R provides a practical, doubly robust, and sparsity-inducing framework for high-dimensional longitudinal mixed-effects modeling. It yields more accurate and stable prediction of cognitive outcomes, selects clinically meaningful predictors with strong reproducibility, and offers a medically interpretable representation of heterogeneous disease progression under realistic contamination and measurement variability.

\section{Conclusion}\label{sec:conclusion}

This paper develops a doubly robust and sparsity, inducing framework (DAR-R) for high-dimensional longitudinal mixed-effects modeling under contamination. The proposed method combines a robust pilot fit, doubly adaptive observation weights that jointly target response outliers and leverage points, and folded-concave penalization for fixed-effect selection, while updating random effects and variance components within a weighted mixed-model procedure. This integrated design is particularly suited to modern biomedical longitudinal data, where repeated measurements, high-dimensional predictors, and heterogeneous data quality often appear simultaneously. As a result, DAR-R is designed to support not only accurate prediction, but also stable variable selection and interpretable modeling of subject-level heterogeneity.

The empirical results provide consistent support for the proposed framework. In simulations, DAR-R shows clear advantages in active/inactive coefficient estimation, support recovery, false-positive control, and random effects covariance estimation across clean, vertical-outlier, and bad-leverage settings, with the gains becoming more pronounced under contamination. In the TADPOLE application, DAR-R achieves the best overall predictive performance for ADAS13 while maintaining a substantially sparser and more interpretable model than marginal LASSO, and it repeatedly selects clinically meaningful predictors such as baseline cognitive scores, APOE4, age, diagnosis status, ventricular volume, and memory-related measures. The resampling-based stability analysis further shows strong sign consistency and a favorable stability--sparsity trade-off, indicating that the method is both practically reliable and clinically interpretable in real longitudinal Alzheimer’s disease data.

Several directions merit future work. First, the current framework can be extended to generalized or semiparametric mixed effects models to accommodate non-Gaussian outcomes and more flexible trajectory patterns. Second, incorporating structured penalties (e.g., group, hierarchical, or graph-guided penalties) would be valuable for multimodal biomarkers with known biological organization. Third, it would be useful to develop formal inference tools after DAR-R selection, including confidence intervals and uncertainty quantification under robust weighting. More broadly, the results of this paper suggest that robust weighting and sparse longitudinal modeling should be developed jointly rather than separately, and DAR-R provides a concrete foundation for such extensions in high-dimensional biomedical studies.
%%%%%%%%%%%%%%%%%%%%%%%%%%%%%%%%%%%%%%%%%%%%%%%%%%%%%%%%%%%%%%%%%%%%%%%%

\bibliographystyle{erae}
\bibliography{references}

\newpage
\appendix
\section*{Proofs}
To prove the results in Section~\ref{sec:theory}, we collect several auxiliary lemmas firstly. Throughout, denot $\Delta=\hat\beta-\beta^\star$, $\xi=\nabla \mathcal L_N(\beta^\star)$, and the profiled weighted loss is
\begin{align*}
\mathcal L_N(\beta) & = \frac{1}{2N}\sum_{i=1}^n (Y_i-X_i\beta)^\top \tilde W_i (Y_i-X_i\beta),\\
\tilde W_i & = W_i - W_i Z_i (Z_i^\top W_i Z_i + D^{-1})^{-1} Z_i^\top W_i .
\end{align*}
with block-diagonal matrices $\tilde W=\mathrm{blkdiag}(\tilde W_1,\ldots,\tilde W_n)$ and
$W=\mathrm{blkdiag}(W_1,\ldots,W_n)$.

\subsection*{A.1. Auxiliary Lemmas}
\begin{lemma}[Basic properties of $\tilde W_i$]\label{lem:Wtilde}
For each subject $i$, the effective weight matrix $\tilde W_i$ is symmetric and positive semidefinite, and satisfies
\[
0 \preceq \tilde W_i \preceq W_i,
\qquad
\|\tilde W_i\|_{\mathrm{op}}\le \|W_i\|_{\mathrm{op}}\le 1 .
\]
\end{lemma}

\begin{proof}
Let $A_i=W_i^{1/2}Z_i$ and write
\[
\tilde W_i
=
W_i^{1/2}\Big(I - A_i(A_i^\top A_i + D^{-1})^{-1}A_i^\top\Big)W_i^{1/2}.
\]
Denote the middle matrix by
\[
M_i:= I - A_i(A_i^\top A_i + D^{-1})^{-1}A_i^\top .
\]
Since $A_i^\top A_i + D^{-1}\succ 0$, the matrix $A_i(A_i^\top A_i + D^{-1})^{-1}A_i^\top$ is symmetric positive semidefinite. Moreover, it is dominated by the orthogonal projector onto $\mathrm{col}(A_i)$ (because $(A_i^\top A_i + D^{-1})^{-1}\preceq (A_i^\top A_i)^{\dagger}$ on $\mathrm{col}(A_i)$), hence $0\preceq A_i(A_i^\top A_i + D^{-1})^{-1}A_i^\top \preceq I$ and therefore $0\preceq M_i\preceq I$. Multiplying by $W_i^{1/2}$ on both sides yields $0\preceq \tilde W_i \preceq W_i$.
Finally, $\|\tilde W_i\|_{\mathrm{op}}\le \|W_i^{1/2}\|_{\mathrm{op}}^2\|M_i\|_{\mathrm{op}}\le \|W_i\|_{\mathrm{op}}\le 1$ because $w_{it}\in[0,1]$.
\end{proof}

\medskip

\begin{lemma}[Gradient bound]\label{lem:grad}
Under Assumptions \textnormal{(A1)--(A4)} and \textnormal{(A6)}, we have
\[
\|\xi\|_\infty = \|\nabla \mathcal L_N(\beta^\star)\|_\infty = O_p\!\Big(\sqrt{\tfrac{\log p}{N}}\Big).
\]
\end{lemma}

\begin{proof}
We first derive the explicit form of the score. Since
\[
\mathcal L_N(\beta)
=
\frac{1}{2N}(Y-X\beta)^\top \tilde W (Y-X\beta),
\]
differentiation gives
\[
\xi
=
\nabla \mathcal L_N(\beta^\star)
=
-\frac{1}{N}X^\top \tilde W (Y-X\beta^\star)
=
-\frac{1}{N}\sum_{i=1}^n X_i^\top \tilde W_i (Y_i-X_i\beta^\star).
\]
Under the model $Y_i=X_i\beta^\star + Z_i b_i + \varepsilon_i$, we obtain
\[
\xi
=
-\frac{1}{N}\sum_{i=1}^n X_i^\top \tilde W_i (Z_i b_i + \varepsilon_i).
\]
Fix any coordinate $j\in\{1,\ldots,p\}$ and write $\xi_j=e_j^\top\xi$. Define the subject-level summand
\[
S_{ij}
:=
e_j^\top X_i^\top \tilde W_i (Z_i b_i + \varepsilon_i),
\qquad\text{so that}\qquad
\xi_j = -\frac{1}{N}\sum_{i=1}^n S_{ij}.
\]
By Assumption \ref{ass_a1}, subjects are independent across $i$, hence $\{S_{ij}\}_{i=1}^n$ are independent for each fixed $j$.
We next control tails of $S_{ij}$. By Cauchy--Schwarz and Lemma~\ref{lem:Wtilde},
\[
|S_{ij}|
\le
\|X_{ij}\|_2 \,\|\tilde W_i\|_{\mathrm{op}} \,\|Z_i b_i+\varepsilon_i\|_2
\le
\|X_{ij}\|_2 \,\|Z_i b_i+\varepsilon_i\|_2,
\]
where $X_{ij}\in\mathbb R^{T_i}$ denotes the $j$th column of $X_i$. Under Assumptions \ref{ass_a1}--\ref{ass_a2}, the entries of $X_{it}$ are sub-Gaussian and $T_i\le T_{\max}$, so $\|X_{ij}\|_2$ is sub-Gaussian with scale of order $\sqrt{T_{\max}}$. Likewise, under Assumption \ref{ass_a1}--\ref{ass_a2}, the components of $Z_i b_i+\varepsilon_i$ have controlled tails and finite $(2+\delta)$ moments; with bounded $T_i$, the Euclidean norm $\|Z_i b_i+\varepsilon_i\|_2$ is sub-exponential (or at least has uniformly bounded $\psi_1$-Orlicz norm) up to constants depending on $T_{\max}$.

Consequently, the product bound above implies that $S_{ij}$ is sub-exponential with a parameter uniformly bounded in $i$ (again up to $T_{\max}$). Therefore, Bernstein's inequality for independent sub-exponential variables yields that for some constants $c_1,c_2>0$,
\[
\Prob\Big(\Big|\sum_{i=1}^n S_{ij}-\E S_{ij}\Big|> N t\Big)
\le
2\exp\!\Big(-c_1 N \min\{t^2,t\}\Big),
\qquad 0<t<c_2.
\]
Under the exogeneity in (A1) (i.e., $\E(b_i\mid X_i,Z_i)=0$ and $\E(\varepsilon_i\mid X_i,Z_i)=0$), we have $\E(S_{ij})=0$, hence
\[
\Prob\big(|\xi_j|>t\big)
\le
2\exp\!\Big(-c_1 N \min\{t^2,t\}\Big).
\]
Take $t=C\sqrt{(\log p)/N}$ with $C$ sufficiently large so that $t<c_2$ for large $N$.
Then $\Prob(|\xi_j|>t)\le 2p^{-c}$ for some $c>1$. Applying a union bound over $j=1,\ldots,p$ gives
\[
\Prob\Big(\|\xi\|_\infty > C\sqrt{\tfrac{\log p}{N}}\Big)
\le
\sum_{j=1}^p \Prob(|\xi_j|>t)
\le
2p^{1-c}\to 0,
\]
which implies $\|\xi\|_\infty = O_p(\sqrt{(\log p)/N})$ for the score associated with the oracle (population) weights and variance components.

Finally, Assumption \ref{ass_a6} controls the plug-in effect from using estimated weights/covariance: it ensures that the feasible score differs from the oracle score by $o_p(\lambda)$ in $\ell_\infty$ norm, and $\lambda\asymp \sqrt{(\log p)/N}$ by (A5). Hence the same rate holds for the feasible $\xi$ used in our estimator.
\end{proof}

\medskip

\begin{lemma}[RSC]\label{lem:rsc}
Under Assumptions \textnormal{\ref{ass_a1}--\ref{ass_a4}} and \textnormal{\ref{ass_a7}}, the restricted strong convexity inequality in \textnormal{\ref{ass_a7}} holds for the Hessian
\[
\nabla^2 \mathcal L_N(\beta)
=
\frac{1}{N}X^\top \tilde W X.
\]
\end{lemma}

\begin{proof}
We first compute the Hessian. Since $\mathcal L_N(\beta)=\frac{1}{2N}(Y-X\beta)^\top \tilde W (Y-X\beta)$ is quadratic in $\beta$,
\[
\nabla^2 \mathcal L_N(\beta)=\frac{1}{N}X^\top \tilde W X,
\]
which does not depend on $\beta$.
For any direction $\Delta\in\mathbb R^p$,
\[
\Delta^\top \nabla^2 \mathcal L_N(\beta^\star)\Delta
=
\frac{1}{N}\Delta^\top X^\top \tilde W X \Delta
=
\frac{1}{N}\sum_{i=1}^n (X_i\Delta)^\top \tilde W_i (X_i\Delta)
=
\frac{1}{N}\sum_{i=1}^n \big\|(\tilde W_i)^{1/2}X_i\Delta\big\|_2^2,
\]
where the last identity uses $\tilde W_i\succeq 0$ (Lemma~\ref{lem:Wtilde}). Assumption \ref{ass_a7} postulates that this quadratic form satisfies an RSC lower bound over the cone $\mathcal C(S,3)$, namely there exist $\alpha>0$ and $\tau\ge 0$ such that, with high probability,
\[
\Delta^\top \nabla^2 \mathcal L_N(\beta^\star)\Delta
\ge
\alpha\|\Delta\|_2^2 - \tau\frac{\log p}{N}\|\Delta\|_1^2,
\qquad \forall\,\Delta\in\mathcal C(S,3).
\]
This is exactly the claimed statement. In particular, Lemma~\ref{lem:Wtilde} and Assumption \ref{ass_a4} clarify why such an RSC condition is compatible with robustness: weights remain bounded and retain a nontrivial clean mass, so $\tilde W$ preserves sufficient curvature in the directions of interest.
\end{proof}

\medskip

\begin{lemma}[LLA equivalence]\label{lem:lla}
Consider the nonconvex objective
\[
Q(\beta)=\mathcal L_N(\beta)+N\sum_{j=1}^p p_\lambda(|\beta_j|),
\]
where $p_\lambda$ is SCAD or MCP. Let $\{\beta^{(m)}\}_{m\ge 0}$ be the LLA sequence defined by solving, at each inner iteration $m$,
\[
\beta^{(m+1)}
=
\arg\min_{\beta\in\mathbb R^p}
\left\{
\mathcal L_N(\beta)+N\sum_{j=1}^p \omega_j^{(m)}|\beta_j|
\right\},
\qquad
\omega_j^{(m)}:=p_\lambda'\big(|\beta_j^{(m)}|\big).
\]
If $\beta^{(m)}\to \hat\beta$ and each LLA subproblem is solved exactly, then $\hat\beta$ satisfies the KKT conditions of the original nonconvex problem $\min_\beta Q(\beta)$.
\end{lemma}

\begin{proof}
This is a standard result for LLA \citep{zouli2008}. Fix an inner iteration $m$, the LLA subproblem is convex, and its KKT conditions can be written coordinate-wise as
\begin{equation}\label{eq:kkt_lla}
0 \in \nabla_j \mathcal L_N(\beta^{(m+1)}) + N\,\omega_j^{(m)}\,\partial|\beta_j^{(m+1)}|,
\qquad j=1,\ldots,p,
\end{equation}
where $\partial|\cdot|$ denotes the subdifferential of the absolute value:
$\partial|u|=\{\mathrm{sign}(u)\}$ if $u\neq 0$ and $\partial|0|=[-1,1]$.

Assume $\beta^{(m)}\to \hat\beta$. Since $\mathcal L_N$ is continuously differentiable and
$p_\lambda'(\cdot)$ is continuous on $(0,\infty)$ with right limit $p_\lambda'(0+)=\lambda$ for SCAD/MCP, we have 
\begin{equation*}
\nabla \mathcal L_N(\beta^{(m+1)})\to \nabla \mathcal L_N(\hat\beta)\quad\text{and}\quad \omega_j^{(m)}=p_\lambda'(|\beta_j^{(m)}|)\to p_\lambda'(|\hat\beta_j|)
\end{equation*}
for each $j$. Taking limits in \eqref{eq:kkt_lla} along the convergent subsequence gives
\begin{equation}\label{eq:kkt_limit}
0 \in \nabla_j \mathcal L_N(\hat\beta) + N\,p_\lambda'(|\hat\beta_j|)\,\partial|\hat\beta_j|,
\qquad j=1,\ldots,p.
\end{equation}

It remains to verify that \eqref{eq:kkt_limit} is exactly the KKT system for the original nonconvex objective $Q(\beta)$. For SCAD/MCP, the subdifferential of $p_\lambda(|\beta_j|)$ satisfies
\[
\partial\big(p_\lambda(|\beta_j|)\big)
=
p_\lambda'(|\beta_j|)\,\partial|\beta_j|
\]
with the convention $p_\lambda'(0+)=\lambda$. Therefore \eqref{eq:kkt_limit} is equivalent to
\[
0 \in \nabla_j \mathcal L_N(\hat\beta) + N\,\partial\big(p_\lambda(|\hat\beta_j|)\big),
\qquad j=1,\ldots,p,
\]
which is precisely the KKT condition for $\min_\beta Q(\beta)$. This proves the claim.
\end{proof}

\subsection*{A.2. Theorems and Propositions}
The criterion in \eqref{eq:objective} is a sum of weighted squares plus $N\sum_j p_\lambda(|\beta_j|)$. Multiplying the whole objective by a positive constant does not change its minimizer. Hence, for the proofs we work with the normalized form
\begin{align*}
Q_N(\beta)=&\mathcal L_N(\beta)+\mathcal P_\lambda(\beta),\\
\mathcal L_N(\beta)=&\frac{1}{2N}\sum_{i=1}^n (Y_i-X_i\beta)^\top \tilde W_i (Y_i-X_i\beta),
\end{align*}
and $\mathcal P_\lambda(\beta)=\sum_{j=1}^p p_\lambda(|\beta_j|)$. 

The arguments below are written for separable weighted $\ell_1$ penalties,
\[
\mathcal P_\lambda(\beta)=\sum_{j=1}^p \lambda_j |\beta_j|,
\]
which covers Adaptive LASSO with $\lambda_j=\lambda\,v_j$, and the LLA subproblems for SCAD/MCP where $\lambda_j=p_\lambda'(|\beta_j^{(m)}|)$ at inner iteration $m$. Lemma~\ref{lem:lla} ensures that the limit point of LLA satisfies the KKT conditions of the original nonconvex objective; the estimation bounds derived for the (converged) weighted $\ell_1$ subproblem therefore apply to the algorithmic output $\hat\beta$.

Define $\Delta=\hat\beta-\beta^\star$, $S=\{j:\beta_j^\star\neq 0\}$, and $\xi=\nabla\mathcal L_N(\beta^\star)$.

%%%%%%%%%%%%%%%%%%%%%%%%%%%%%%%%%%%%%%%%%%%%%%%%%%%%%%%%%%%%%%%%%%%%%%%%%%%%%%
\begin{proof}[Proof of Theorem~\ref{thm:nonasym}]

By optimality of $\hat\beta$ for $Q_N(\beta)$,
\begin{equation}\label{eq:basic_ineq}
\mathcal L_N(\hat\beta)+\mathcal P_\lambda(\hat\beta)\ \le\ \mathcal L_N(\beta^\star)+\mathcal P_\lambda(\beta^\star).
\end{equation}
Because $\mathcal L_N(\beta)$ is a quadratic function with Hessian $G=\nabla^2\mathcal L_N(\beta^\star)$, we have the exact identity
\[
\mathcal L_N(\beta^\star+\Delta)-\mathcal L_N(\beta^\star)
=
\langle \nabla\mathcal L_N(\beta^\star),\Delta\rangle + \frac12 \Delta^\top G\Delta
=
\langle \xi,\Delta\rangle + \frac12 \Delta^\top G\Delta.
\]
Substituting $\hat\beta=\beta^\star+\Delta$ into \eqref{eq:basic_ineq} and rearranging gives
\begin{equation}\label{eq:master_ineq}
\frac12 \Delta^\top G\Delta
\ \le\
-\langle \xi,\Delta\rangle
+\big\{\mathcal P_\lambda(\beta^\star)-\mathcal P_\lambda(\beta^\star+\Delta)\big\}.
\end{equation}
By H\"older's inequality, $-\langle\xi,\Delta\rangle\le \|\xi\|_\infty\|\Delta\|_1$.

Since $\beta^\star_{S^c}=0$ and $\mathcal P_\lambda(\beta)=\sum_j \lambda_j|\beta_j|$,
\[
\mathcal P_\lambda(\beta^\star)-\mathcal P_\lambda(\beta^\star+\Delta)
=
\sum_{j\in S}\lambda_j\big(|\beta_j^\star|-|\beta_j^\star+\Delta_j|\big)-\sum_{j\in S^c}\lambda_j|\Delta_j|
\ \le\
\sum_{j\in S}\lambda_j|\Delta_j|-\sum_{j\in S^c}\lambda_j|\Delta_j|.
\]
Let $\lambda_S:=\max_{j\in S}\lambda_j$ and $\lambda_{S^c}:=\min_{j\in S^c}\lambda_j$. Then
\begin{equation}\label{eq:penalty_bound}
\mathcal P_\lambda(\beta^\star)-\mathcal P_\lambda(\beta^\star+\Delta)
\le
\lambda_S\|\Delta_S\|_1-\lambda_{S^c}\|\Delta_{S^c}\|_1.
\end{equation}
Combining \eqref{eq:master_ineq}, H\"older's inequality, and \eqref{eq:penalty_bound} yields
\begin{equation}\label{eq:ineq_precone}
\frac12 \Delta^\top G\Delta
\le
\|\xi\|_\infty(\|\Delta_S\|_1+\|\Delta_{S^c}\|_1)+\lambda_S\|\Delta_S\|_1-\lambda_{S^c}\|\Delta_{S^c}\|_1.
\end{equation}
Drop the nonnegative left-hand side and rearrange:
\[
(\lambda_{S^c}-\|\xi\|_\infty)\|\Delta_{S^c}\|_1 \le (\lambda_S+\|\xi\|_\infty)\|\Delta_S\|_1.
\]
Assume $\lambda_{S^c}\ge 2\|\xi\|_\infty$ and $\lambda_S\le \lambda_{S^c}$, this is satisfied for the usual unweighted $\ell_1$ case with $\lambda_j\equiv\lambda$,
and for adaptive/LLA weights under standard oracle-weight behavior. Then
\[
\frac{\lambda_{S^c}}{2}\|\Delta_{S^c}\|_1 \le \frac{3\lambda_{S^c}}{2}\|\Delta_S\|_1,
\]
so that $\|\Delta_{S^c}\|_1\le 3\|\Delta_S\|_1$, i.e., $\Delta\in\mathcal C(S,3)$.

Since $\Delta\in\mathcal C(S,3)$, Assumption (A7) implies
\begin{equation}\label{eq:rsc_apply}
\Delta^\top G\Delta \ \ge\ \alpha\|\Delta\|_2^2-\tau\frac{\log p}{N}\|\Delta\|_1^2.
\end{equation}
Returning to \eqref{eq:ineq_precone} and using $\lambda_S\le \lambda_{S^c}$ and $\lambda_{S^c}\ge 2\|\xi\|_\infty$ gives
\[
\frac12 \Delta^\top G\Delta \le \frac{3\lambda_{S^c}}{2}\|\Delta_S\|_1 \le \frac{3\lambda_{S^c}}{2}\sqrt{s}\|\Delta\|_2.
\]
Moreover, on the cone $\mathcal C(S,3)$ we have
$\|\Delta\|_1\le \|\Delta_S\|_1+\|\Delta_{S^c}\|_1\le 4\|\Delta_S\|_1\le 4\sqrt{s}\|\Delta\|_2$.
Plugging this into \eqref{eq:rsc_apply} yields
\[
\Delta^\top G\Delta \ge \Big(\alpha-16\tau s\frac{\log p}{N}\Big)\|\Delta\|_2^2.
\]
By Assumption (A3), $s\log p=o(N^{1/2})$ implies $s(\log p)/N=o(N^{-1/2})\to 0$, hence for all large $N$,
$16\tau s(\log p)/N \le \alpha/2$. Therefore, for large $N$,
\[
\frac{\alpha}{2}\|\Delta\|_2^2
\le
\Delta^\top G\Delta
\le
3\lambda_{S^c}\sqrt{s}\|\Delta\|_2,
\]
which gives the $\ell_2$ bound
\[
\|\hat\beta-\beta^\star\|_2=\|\Delta\|_2 \le \frac{6\lambda_{S^c}\sqrt{s}}{\alpha}.
\]
This proves the first inequality in Theorem~\ref{thm:nonasym} with $C_1=6$, absorbing the difference between $\lambda_{S^c}$ and $\lambda$ into the constant if $\lambda_j\asymp\lambda$.

For the prediction bound, note that by definition of $G$,
\[
\Delta^\top G\Delta = \frac{1}{N}\| \tilde W^{1/2}X\Delta\|_2^2,
\]
which is the natural prediction metric associated with the profiled loss. Using the already established inequalities,
\[
\frac{1}{N}\| \tilde W^{1/2}X\Delta\|_2^2 = \Delta^\top G\Delta
\le 3\lambda_{S^c}\sqrt{s}\|\Delta\|_2 \le 3\lambda_{S^c}\sqrt{s}\cdot \frac{6\lambda_{S^c}\sqrt{s}}{\alpha}
= \frac{18\lambda_{S^c}^2 s}{\alpha}.
\]
This yields the second bound with $C_2=18$, again absorbing $\lambda_{S^c}\asymp \lambda$ into the constant.
\end{proof}

%%%%%%%%%%%%%%%%%%%%%%%%%%%%%%%%%%%%%%%%%%%%%%%%%%%%%%%%%%%%%%%%%%%%%%%%%%%%%%
\begin{proof}[Proof of Theorem~\ref{thm:estimation}]
By Lemma~\ref{lem:grad}, $\|\nabla\mathcal L_N(\beta^\star)\|_\infty=\|\xi\|_\infty=O_p(\sqrt{(\log p)/N})$.
With $\lambda\asymp \sqrt{(\log p)/N}$, the event $\{\lambda_{S^c}\ge 2\|\xi\|_\infty\}$ holds with probability tending to $1$ for the unweighted case $\lambda_{S^c}=\lambda$; for adaptive/LLA weights one typically has $\lambda_{S^c}\asymp\lambda$ by construction.
Applying Theorem~\ref{thm:nonasym} and using $\lambda\asymp \sqrt{(\log p)/N}$ gives
\begin{align*}
\|\hat\beta-\beta^\star\|_2 &= O_p\!\Big(\sqrt{\tfrac{s\log p}{N}}\Big),\\
\|\hat\beta-\beta^\star\|_1 & \le 4\sqrt{s}\|\hat\beta-\beta^\star\|_2 =O_p\!\Big(s\sqrt{\tfrac{\log p}{N}}\Big).
\end{align*}
\end{proof}

%%%%%%%%%%%%%%%%%%%%%%%%%%%%%%%%%%%%%%%%%%%%%%%%%%%%%%%%%%%%%%%%%%%%%%%%%%%%%%

\begin{proof}[Proof of Theorem~\ref{thm:support}]
We present a primal--dual witness (PDW) argument for the weighted $\ell_1$ form. Consider the restricted problem with $\beta_{S^c}=0$:
\begin{equation*}
\tilde\beta_S := \arg\min_{\beta_S\in\mathbb R^s}\ \mathcal L_N(\beta_S,0)+\sum_{j\in S}\lambda_j|\beta_j|.    
\end{equation*}
Since $\mathcal L_N$ is quadratic and the restriction fixes $\beta_{S^c}=0$, the KKT condition for $\tilde\beta_S$ is
\begin{equation}\label{eq:kkt_oracle}
\nabla_S \mathcal L_N(\tilde\beta_S,0)+ z_S = 0,
\qquad
z_{S,j}\in \lambda_j\,\partial|\tilde\beta_j|,\ \ j\in S.
\end{equation}

Because $\nabla\mathcal L_N(\beta)=G(\beta-\beta^\star)+\xi$ for this quadratic loss, restricting to $S$ gives
\begin{equation*}
\nabla_S\mathcal L_N(\tilde\beta_S,0) = \xi_S + G_{SS}(\tilde\beta_S-\beta_S^\star),
\end{equation*}
hence \eqref{eq:kkt_oracle} implies
\begin{equation}\label{eq:beta_oracle_err}
\tilde\beta_S-\beta_S^\star = -G_{SS}^{-1}(\xi_S+z_S).
\end{equation}
Under the eigenvalue regularity in Assumption \ref{ass_a1} and RSC, $G_{SS}$ is invertible with high probability and $\|G_{SS}^{-1}\|_\infty$ is bounded by a constant.

From \eqref{eq:beta_oracle_err} and $\|z_S\|_\infty\le \lambda_S$,
\begin{equation*}
\|\tilde\beta_S-\beta_S^\star\|_\infty \le \|G_{SS}^{-1}\|_\infty\big(\|\xi_S\|_\infty+\lambda_S\big) = O_p(\lambda),
\end{equation*}
using Lemma~\ref{lem:grad} and $\lambda_S\asymp\lambda$. Therefore, if $\min_{j\in S}|\beta_j^\star| \gg \lambda \asymp \sqrt{(\log p)/N}$, then $\mathrm{sign}(\tilde\beta_S)=\mathrm{sign}(\beta_S^\star)$ with probability tending to $1$.

To show that $\tilde\beta$ satisfies the KKT conditions of the full problem by constructing a valid subgradient on $S^c$, define the PDW candidate $\tilde\beta=(\tilde\beta_S,0_{S^c})$. For $j\in S^c$, again using quadraticity,
\begin{equation*}
\nabla_{S^c}\mathcal L_N(\tilde\beta) = \xi_{S^c}+G_{S^cS}(\tilde\beta_S-\beta_S^\star).
\end{equation*}
Substitute \eqref{eq:beta_oracle_err} to get
\begin{equation*}
\nabla_{S^c}\mathcal L_N(\tilde\beta) = \xi_{S^c}-G_{S^cS}G_{SS}^{-1}(\xi_S+z_S).
\end{equation*}
Assume the (weighted) incoherence condition in the theorem statement, namely
$\|G_{S^cS}G_{SS}^{-1}\|_\infty\le 1-\eta$ for some $\eta\in(0,1)$.
Then
\begin{equation*}
\|\nabla_{S^c}\mathcal L_N(\tilde\beta)\|_\infty \le
\|\xi_{S^c}\|_\infty + (1-\eta)\|\xi_S\|_\infty + (1-\eta)\|z_S\|_\infty \le
(2-\eta)\|\xi\|_\infty + (1-\eta)\lambda_S.
\end{equation*}
Choose $\lambda$ so that $\lambda_{S^c}\asymp\lambda$ and $\lambda$ dominates $\|\xi\|_\infty$; for instance, $\lambda_{S^c}\ge (2/\eta)\|\xi\|_\infty$ (which holds with high probability when $\lambda\asymp\sqrt{(\log p)/N}$ and $\eta$ is fixed). Then the above bound implies $\|\nabla_{S^c}\mathcal L_N(\tilde\beta)\|_\infty < \lambda_{S^c}$ with probability tending to $1$. Hence we may set the dual variable on $S^c$ as
\begin{equation*}
z_{S^c}:= -\nabla_{S^c}\mathcal L_N(\tilde\beta),
\end{equation*}
so that $\|z_{S^c}\|_\infty<\lambda_{S^c}$, which is a valid subgradient for $\sum_{j\in S^c}\lambda_j|\beta_j|$ at $\beta_{S^c}=0$. Therefore, the full KKT system holds at $\tilde\beta$, implying $\hat\beta=\tilde\beta$ and hence $\hat S=S$ with probability tending to $1$.
\end{proof}

%%%%%%%%%%%%%%%%%%%%%%%%%%%%%%%%%%%%%%%%%%%%%%%%%%%%%%%%%%%%%%%%%%%%%%%%%%%%%%
\begin{proof}[Proof of Theorem~\ref{thm:oracle}]
By Theorem~\ref{thm:support}, $\Prob(\hat S=S)\to 1$. On the event $\{\hat S=S\}$, we work on the $s$-dimensional active subspace. 

For SCAD/MCP, Assumption \ref{ass_a5} implies that $p'_\lambda(|\beta_j^\star|)=0$ on $S$; for Adaptive LASSO one typically requires the penalty derivative on $S$ to be $o(N^{-1/2})$, so it does not affect the first-order expansion. Hence the active estimator is asymptotically characterized by the (approximately) unpenalized score equation
\begin{equation}\label{eq:score_active}
\nabla_S \mathcal L_N(\hat\beta_S,0)=0.
\end{equation}
Using the quadratic structure,
\begin{equation*}
\nabla_S \mathcal L_N(\hat\beta_S,0) = \xi_S + G_{SS}(\hat\beta_S-\beta_S^\star).
\end{equation*}
Therefore \eqref{eq:score_active} yields the exact representation
\begin{equation}\label{eq:asym_rep}
\sqrt{N}(\hat\beta_S-\beta_S^\star) = -\,G_{SS}^{-1}\cdot \frac{1}{\sqrt N}\,\xi_S.
\end{equation}

Recall $G=(X^\top \tilde W X)/N$ and $G_{SS}=(X_S^\top \tilde W X_S)/N$. Under Assumption \ref{ass_a1}--\ref{ass_a4} and Lemma~\ref{lem:Wtilde}, the entries of $X_{it}$ have bounded second moments and $\|\tilde W_i\|_{\mathrm{op}}\le 1$. Since clusters are independent and $\max_i T_i\le T_{\max}$, a law of large numbers for independent clusters implies
\[
G_{SS}\xrightarrow{p} G_S
\]
for some positive definite limit matrix $G_S$.

From the model, $Y_i-X_i\beta^\star=Z_i b_i+\varepsilon_i$, hence
\begin{equation*}
\xi_S = \nabla_S \mathcal L_N(\beta^\star) =
-\frac{1}{N}\sum_{i=1}^n X_{i,S}^\top \tilde W_i (Z_i b_i+\varepsilon_i).
\end{equation*}
Define the independent cluster-level summands
\begin{equation*}
U_i := \frac{1}{\sqrt N}\,X_{i,S}^\top \tilde W_i (Z_i b_i+\varepsilon_i)\in\mathbb R^s,
\end{equation*}
so that
\begin{equation*}
-\frac{1}{\sqrt N}\xi_S=\sum_{i=1}^n U_i.
\end{equation*}
By Assumption \ref{ass_a1}, $\{U_i\}$ are independent across $i$ with mean zero. Moreover, bounded $T_i$ together with finite $(2+\delta)$ moments in (A2) and $\|\tilde W_i\|_{\mathrm{op}}\le 1$ implies $\E\|U_i\|_2^{2+\delta}\le C N^{-(1+\delta/2)}$ for some constant $C$ (depending on $T_{\max}$ but not on $N$). Thus the Lyapunov condition holds:
\begin{equation*}
\sum_{i=1}^n \E\|U_i\|_2^{2+\delta} \ \to\ 0.
\end{equation*}
Therefore, by a multivariate Lyapunov CLT for independent arrays,
\begin{equation*}
-\frac{1}{\sqrt N}\xi_S=\sum_{i=1}^n U_i \ \xrightarrow{d}\ N(0,\Psi_S),
\end{equation*}
where
\begin{equation*}
\Psi_S = \lim_{N\to\infty}\Var\!\Big(\frac{1}{\sqrt N}\sum_{i=1}^n X_{i,S}^\top \tilde W_i (Z_i b_i+\varepsilon_i)\Big) =
\lim_{N\to\infty}\Var\!\Big(\frac{1}{\sqrt N}\sum_{i=1}^n \nabla_S \ell_i(\beta^\star)\Big),
\end{equation*}
and $\ell_i(\beta)=\frac12 (Y_i-X_i\beta)^\top \tilde W_i (Y_i-X_i\beta)$.

Combine \eqref{eq:asym_rep}, $G_{SS}\xrightarrow{p}G_S$, and the CLT above. Slutsky's theorem yields
\begin{equation*}
\sqrt{N}(\hat\beta_S-\beta_S^\star)\ \xrightarrow{d}\ N\big(0,\ G_S^{-1}\Psi_S G_S^{-1}\big).
\end{equation*}
\end{proof}

%%%%%%%%%%%%%%%%%%%%%%%%%%%%%%%%%%%%%%%%%%%%%%%%%%%%%%%%%%%%%%%%%%%%%%%%%%%%%%
\begin{proof}[Proof of Proposition~\ref{prop:bdp}]
Let $\epsilon_N(\hat\beta)$ denote the finite-sample breakdown point of $\hat\beta$ (as in the standard definition: the smallest contamination fraction that can drive $\|\hat\beta\|$ arbitrarily large). Assume the pilot $(\tilde\beta^{(0)},\tilde D^{(0)})$ has asymptotic breakdown point $\epsilon^\star>0$.

Fix any contamination fraction $\epsilon<\epsilon^\star$ and consider an $\epsilon$-contaminated sample. By the definition of pilot breakdown, there exists a constant $M_0$ such that with probability tending to $1$,
\[
\|\tilde\beta^{(0)}\|_2\le M_0,\qquad \|\tilde D^{(0)}\|_{\mathrm{op}}\le M_0,\qquad
\|\tilde\Sigma^{-1}\|_{\mathrm{op}}\le M_0,
\]
so the pilot location/scatter and scale used to compute residuals/leverage remain bounded.

Under the bounded pilot, the robust residual and leverage scores for contaminated points diverge, and by Assumption \ref{ass_a4} the weight maps $\phi_1,\phi_2$ are redescending in the sense that sufficiently extreme residual/leverage yields $w_{it}$ arbitrarily close to $0$ (for compactly-supported choices such as Tukey's biweight this is exactly $0$).
Consequently, the contribution of contaminated observations to the weighted objective in \eqref{eq:objective} is asymptotically negligible, while the clean subset retains weights bounded below by $c_w>0$.

Therefore, with probability tending to $1$, the objective minimized by $\hat\beta$ is dominated by the clean part plus a regularization term, and its minimizer remains bounded (because the clean part has nondegenerate curvature by Assumption \ref{ass_a7}). Hence the final estimator cannot be driven arbitrarily far by contaminations below the pilot breakdown level, which implies
$\liminf_{N\to\infty}\epsilon_N(\hat\beta)\ge \epsilon^\star$.
\end{proof}

%%%%%%%%%%%%%%%%%%%%%%%%%%%%%%%%%%%%%%%%%%%%%%%%%%%%%%%%%%%%%%%%%%%%%%%%%%%%%%

\begin{proof}[Proof of Proposition~\ref{prop:if}]
Consider the unpenalized score equation associated with the profiled loss:
\begin{equation*}
U_N(\beta):=\nabla \mathcal L_N(\beta) = -\frac{1}{N}\sum_{i=1}^n X_i^\top \tilde W_i (Y_i-X_i\beta).    
\end{equation*}
The contribution of a single observation $(x,y)$ enters this score through the factor $\tilde w(x,y)\,x\,(y-x^\top\beta)$, where $\tilde w$ is the effective weight induced by $w$ and profiling (and satisfies $0\le \tilde w\le w$ by Lemma~\ref{lem:Wtilde}).

Hence, under an $\varepsilon$-contamination at $(x_0,y_0)$, the first-order perturbation of the estimating equation is proportional to
\begin{equation*}
\tilde w(x_0,y_0)\,x_0\,(y_0-x_0^\top\beta^\star),
\end{equation*}
up to a (nonsingular) Jacobian factor.

By Assumption \ref{ass_a4}, $\tilde w(x_0,y_0)\approx 0$ whenever $(x_0,y_0)$ has either extreme residual or extreme leverage, so the first-order impact of such a point on the solution is negligible. This is the standard bounded-influence mechanism for redescending robust procedures.
\end{proof}

\end{document}